%% file: LQG_game_v3.tex
\documentclass[final,letterpaper]{siamltex}
\input{presiam}

\newcommand{\bmat}[1]{\begin{bmatrix}#1\end{bmatrix}}
\let\VSF  \mathsf

\usepackage[colorlinks]{hyperref} 

\title{Common Information based Markov Perfect Equilibria for Linear-Gaussian Games with Asymmetric Information}

\author{Abhishek Gupta, Ashutosh Nayyar, C\'{e}dric Langbort, Tamer Ba\c{s}ar
\thanks{Ashutosh Nayyar is with the Department of Electrical Engineering at the University of Southern California, Los Angeles, CA, USA. His email is  {\tt\small ashutosn@usc.edu}. All other authors are with Coordinated Science Laboratory at the University of Illinois at Urbana-Champaign, USA. Email: {\tt\small \{gupta54,langbort,basar1\}@illinois.edu}}%
}

\begin{document}
\maketitle
\begin{abstract}
We consider a class of two-player dynamic stochastic nonzero-sum games where the state transition and observation equations are linear, and the primitive random variables are Gaussian. Each controller acquires possibly different dynamic information about the state process and the other controller's past actions and observations. This leads to a dynamic game of asymmetric information among the controllers. Building on our earlier work on finite games with asymmetric information, we devise an algorithm to compute a Nash equilibrium by using the common information among the controllers. We call such equilibria {\it common information based Markov perfect equilibria} of the game, which can be viewed as a refinement of Nash equilibrium in games with asymmetric information. If the players' cost functions are quadratic, then we show that under certain conditions a unique common information based Markov perfect equilibrium exists. Furthermore, this equilibrium can be computed by solving a sequence of linear equations. We also show through an example that there could be other Nash equilibria in a game of asymmetric information, not corresponding to common information based Markov perfect equilibria.
\end{abstract}

\section{Introduction}
A game models a scenario where multiple strategic controllers (or players) optimize their objective functionals, which depend not only on the self actions but also on the actions of other controllers. In stochastic static games, players observe the realization of some random state of nature, possibly through separate noisy channels, and use such observations to independently determine their actions so that the expected values of their individual cost (or utility) functions are optimized. In a stochastic dynamic game, on the other hand, the players act at multiple time steps, based on observation or measurement of some dynamic process which itself is driven by past actions as well as random quantities, which could again be called random states of nature. What information each player acquires at each stage of the game determines what is called the {\it information structure} of the underlying game. If all the players acquire the same information at each time step, then the dynamic game is said to be a game of 
symmetric information. However, in many real scenarios, the players do not have access to the same information about the underlying state processes and other players' observations and past actions. Such games are known as games with asymmetric information. For example, several problems in economic interactions \cite{harsanyi1967,cole2001,myerson1997}, attacks on cyber-physical systems \cite{gupta2012}, auctions, cryptography, etc. can be modeled as games of asymmetric information among strategic players.

Games with symmetric and/or perfect information have been well studied in the literature; see, for example, \cite{Shapley:1953,Sobel1971, fudenberg1991,Basarbook,Filarbook}. In these games, the players have the same beliefs on the states of the game, future observations and future expected costs or payoffs. However, in games with asymmetric information, the players need not have the same beliefs on the current state and future evolution of the game. General frameworks to compute or refine Nash equilibria in stochastic games of symmetric or perfect information have received attention from several researchers, see for example, \cite{maskin2001, myerson1997, fudenberg1991} among many others. However, by comparison, such general frameworks for games of asymmetric information are scant (for exceptions, see \cite{Basar1975,Basaronestep,Basarmulti, basar1985}). This paper, in addition to its earlier finite-game version \cite{nayyar2012a}, provides such a 
framework.

%  Harsanyi \cite{harsanyi1967} proposed a solution concept called ``Bayesian Nash equilibrium'' for static games with asymmetric information, if the beliefs of all the controllers on other controllers' information given their own information are consistent. 

In our recent work \cite{nayyar2012a}, we considered a general finite non-zero sum dynamic stochastic game of asymmetric information with stagewise additive cost functions. Under certain assumptions on the information structures of the players, we obtained a characterization of a particular class of Nash equilibria using a dynamic programming like approach. The key idea there was to use common information among the players to transform the original game of asymmetric information to a game of symmetric and perfect state information with an expanded state and action spaces of the players so that a more easily computable Markov perfect equilibrium of the latter can be used to obtain a Nash equilibrium for the former. The advantage of this technique is that instead of searching for equilibrium in the (large) space of strategies (which grows with the number of stages), we only need to compute Nash equilibrium in a succession of static games of complete information. This reduces the computational effort in 
computing a Nash equilibrium of the game. We call the Nash equilibria obtained with this approach as {\it common information based Markov perfect equilibrium}.

In this work, we extend the framework and results of \cite{nayyar2012a} to infinite games, particularly those with linear state and observation equations and Gaussian random variables. For quadratic cost functions of the players satisfying certain assumptions, we show that a unique common information based Markov perfect equilibrium exists. The general framework developed in the paper can be applied to obtain Nash equilibria of broader classes of stochastic dynamic games with asymmetric information, satisfying two general assumptions delineated in the paper.  

\subsection{Previous Work}
In the past, specific models of various classes of games have been studied, where different players acquire different information. Harsanyi, in his seminal paper \cite{harsanyi1967}, studied one subclass of static games with finite state and action spaces of the players, and showed that under some technical conditions, Nash equilibrium exists in such games. Various authors \cite{Behn68,Rhodes69,Willman69} have studied two-player zero-sum differential games with linear state dynamics and quadratic payoffs, where the players do not make the same measurements about the state. A zero sum differential game where one player's observation is nested in the other player's observation was considered in \cite{Hojota}. A zero-sum differential game where one player makes a noisy observation of the state while the other one does not make any measurement was considered in \cite{MintzBasar}.

Discrete-time non-zero sum LQG games with one step delayed sharing of observations were studied in \cite{Basaronestep} and \cite{Basarmulti}. A game with one-step delayed observation and action sharing among the players was considered in \cite{Altman:2009}. A two-player finite game in which the players do not have access to each other's observations and  control actions was considered in \cite{hespanha2001}, where a necessary and sufficient condition for existence of a Nash equilibrium in terms of two coupled dynamic programs was obtained.

Obtaining equilibrium solutions for  stochastic  games  when players make independent noisy observations of the state and do not  share all of their information (or even when they have access to the same noisy observation as in \cite{Bjota}) has remained a challenge for general classes of games. Identifying classes of games which would lead to tractable solutions or feasible solution methods is therefore an important goal in this area.

\subsection{Contributions of this Paper}
This paper is a sequel to our earlier finite game work in \cite{nayyar2012a}, where we make similar assumptions on the information structures of the players. We study games in which the state, the players' actions and primitive random variables take values in finite-dimensional Euclidean spaces. The state evolution and observation equations are taken to be linear in their arguments and all primitive random variables are assumed to be mutually independent zero-mean Gaussian random variables. We assume that the players have a stagewise additive total cost function. 

We assume that the information structures of the players satisfy two sufficient conditions. For any dynamic game satisfying these assumptions, we show that we can decompose it into several static games using a backward induction algorithm. If there exists a Nash equilibrium for each of those static games, then there exists a common information based Markov perfect equilibrium for the original dynamic game. Furthermore, we present an algorithm that computes the common information based Markov perfect equilibrium in such games, provided that it exists. For games in which the cost functions of the players are quadratic satisfying certain assumption, we show that the static game at each time step admits a unique Nash equilibrium in the class of all Borel measurable strategies of the players at that time step, thereby proving the existence of a unique common information based Markov perfect equilibrium. We also show, by example, that there may be other Nash equilibria of such games that cannot be computed using 
the conceptual method developed in this paper.

To sum up, common information based Markov perfect equilibria constitute a subclass of Nash equilibria of such games, and it can be thought of as a {\it refinement of Nash equilibrium} for games with asymmetric information. However, we do not look into implementation issues of common information based Markov perfect equilibrium in this paper, and we leave this as a topic of further investigation.

\subsection{Notation}
Random variables are denoted by upper case letters and their realizations by the corresponding lower case letters. Random vectors are denoted by upper case bold letters and their realizations by lower case bold letters. Unless otherwise stated, the state, action and observations are assumed to be vector valued. 

Let $\ALP X$ be a set. For a subset $\SF X\subset\ALP X$, we let $\SF X^\complement$ denote the complement of the set $\SF X$. We use $\text{id}_{\ALP X}$ to denote the identity map on the set $\ALP X$. The transpose of a matrix $A$ is denoted by $A^\transpose$. 

% Let $\ALP X$ and $\ALP Y$ be measurable spaces. If $f:\ALP X\rightarrow\ALP Y$ is a measurable function, then $\dom(f)$ and $\range(f)$, respectively, denote the domain space and the range space of the function $f$.

% Let $\VEC x_1$ and $\VEC x_2$ be two vectors in possibly different finite-dimensional Euclidean spaces. We say that a function $f$ is a concatenation function if $f(\VEC x_1,\VEC x_2) = [\VEC x_1^\transpose, \VEC x_2^\transpose]^\transpose$.   

Subscripts are used as time indices and superscripts are used as player/controller indices. Consider $a,b\in\mathbb{N}$. Let $\VEC X_t$ be an element of a finite dimensional Euclidean space $\ALP X_t$ for $a\leq t\leq b$. If $a\leq b$, then we let $\VEC X_{a:b}$ denote the set of vectors $\{\VEC X_a, \VEC X_{a+1}, \dots, \VEC X_b\}$. If $a>b$, then $\VEC X_{a:b}$ is empty. On the other hand, we use $\ALP X_{a:b}$ to denote the product space $\prod_{t=a}^b\ALP X_t$, which is a finite dimensional Euclidean space, with the understanding that $\ALP X_{a:b}=\emptyset$ if $a>b$. We use a similar convention for superscripts.

We use $\mathds{P}\{\cdot\}$ to denote the probability of an event and $\mathds{E}[\cdot]$ to denote the expectation of a random variable. For a collection of functions $\boldsymbol{g}$, the notations $\mathds{P}^{\boldsymbol{g}}\{\cdot\}$ and $\mathds{E}^{\boldsymbol{g}}[\cdot]$ indicate that the probability/expectation depends on the choice of functions in $\boldsymbol{g}$. Similarly, for a probability measure $\pi$, the notation $\mathds{E}^{\pi}[\cdot]$ indicates that the expectation is with respect to the measure $\pi$. The notation $\ind{\VEC x}$ denotes a Dirac measure at the point $\VEC x$. For a set $\ALP X$ and its subset $\SF X$, $1_{\SF X}:\ALP X\rightarrow\{0,1\}$ denotes the indicator function on the set $\SF X$.

Let $\VEC X, \VEC Y$ and $\VEC Z$ be three random variables taking values, respectively, in the spaces $\ALP X,\:\ALP Y$ and $\ALP Z$. Then, $\mathds{P}\{\SF X|\VEC y,\VEC z\}$ denotes the probability of the event $\SF X\subset\ALP X$ given the realizations $\VEC y$ and $\VEC z$ of the random variables $\VEC Y$ and $\VEC Z$. Similarly, $\mathds{E}[\cdot|\VEC y]$ denotes the expected value of a real-valued function $(\cdot)$ given the realization $\VEC y$. We use $\mathds{P}\{d\VEC x,d\VEC y|\VEC z\}$ to denote the conditional probability measure over the space $\ALP X\times\ALP Y$ given a realization $\VEC z$ of another random variable $\VEC Z$.

\subsection{Outline of the Paper}
The paper is organized as follows. In Section \ref{sec:problem}, we formulate the two-player non-zero sum game problem with linear dynamics, linear observation equations, and asymmetric information among the players. We make two assumptions on the information structures of the controllers and an assumption on the admissible strategies of the agents. We also discuss consequences of the assumptions we make on the information structures. In Section \ref{sec:main}, we state the main result of the paper and develop a backward induction algorithm that computes the common information based Markov perfect equilibrium of the game formulated in Section \ref{sec:problem}, provided that it exists. In Section \ref{sec:lqggames}, we specialize the result of Section \ref{sec:main} to LQG games, and show that under further assumptions on cost functions, a unique common information based Markov perfect equilibrium exists in the class of measurable strategies of the players.
% We use the algorithm developed to compute the common information based Markov perfect equilibrium of two examples in Sections \ref{sec:lqgonestep} and \ref{sec:noisyglobal}. 
In Section \ref{sec:information}, we show through an example that there may be other Nash equilibria of a game with asymmetric information, and that using our algorithm, we compute only a subclass of all Nash equilibria. We discuss some implications of our assumptions in Section \ref{sec:discussion}. Finally, we conclude our discussion in Section \ref{sec:conclude} and identify several directions for future research. Proofs of most of the results in the paper are given in appendices.

\section{Problem Formulation}\label{sec:problem}
Let $\VEC X_t$ be the state of a linear system which is controlled by two controllers (players)\footnote{In the paper, we use the term ``controller'' instead of ``player'', because we introduce another set of players in the symmetric information game introduced in the next section.}. At each time step $t$, Controller $i$, $i=1,2$, observes the state through a noisy sensor; this observation is denoted by $\VEC Y^i_t$.  Controller $i$'s action at time step $t$ is denoted by $\VEC U^i_t$. For each Controller $i\in\{1,2\}$ at time $t\in\{1,\ldots,T-1\}$, the state, action and observation spaces are denoted by $\ALP X_t$, $\ALP U^i_t$ and $\ALP Y^i_t$, respectively, and they are assumed to be finite dimensional Euclidean spaces. The dynamics and observation equations are given as
\beq{\label{eq:lqgdynamics} \VEC X_{t+1} &=& A_t\VEC X_{t}+B^1_t\VEC U^1_t+B^2_t\VEC U^2_t+\VEC W_t^0,\\
\label{eq:lqgobserv} \VEC Y^i_t &=& H^i_t\VEC X_t + \VEC W^i_t, \mbox{~~}i=1,2,}
where $\VEC W^i_t$ is a random variable taking values in a finite dimensional Euclidean space denoted by $\ALP W^i_t$ for all $i\in\{0,1,2\}$ and $t\in\{1,\ldots,T-1\}$, and $A_t,B^i_t,H^i_t,\: i\in\{1,2\},\: t\in\{1,\ldots,T-1\}$ are matrices of appropriate dimensions. $\VEC X_1,\VEC  W^{0:2}_{1:T-1}$ are primitive random variables, and they are assumed to be mutually independent and zero-mean Gaussian random vectors.

\subsection{Information Structures of the Controllers}
The information available to each controller at time step $t\in\{1,\ldots,T-1\}$ is a subset of all information generated in the past, that is, $\{\VEC Y^{1:2}_{1:t},\VEC U^{1:2}_{1:t-1}\}$. Let $\ALP E^i_t$ and $\ALP F^i_t$, respectively, be defined as
\beqq{\ALP E^i_t:=\{(j,s)\in\{1,2\}\times\{1,\ldots,T\}:\text{ Controller } i \text{ at time } t \text{ knows } \VEC Y^j_s\},\\
\ALP F^i_t:=\{(j,s)\in\{1,2\}\times\{1,\ldots,T\}:\text{ Controller } i \text{ at time } t \text{ knows } \VEC U^j_s\}.}
Define $\ALP I^i_t$, $\ALP C_t$ and $\ALP P^i_t$ for $i=1,2$ and $t\in \{1,\ldots,T-1\}$ as
\beqq{\ALP I^i_t  &=& \prod_{(j,s)\in\ALP E^i_t} \ALP Y^j_s\times \prod_{(j,s)\in\ALP F^i_t} \ALP U^j_s,\\
\ALP C_t &=& \prod_{(j,s)\in\ALP E^1_t\cap\ALP E^2_t} \ALP Y^j_s\times \prod_{(j,s)\in\ALP F^1_t\cap \ALP F^2_t} \ALP U^j_s,\\
\ALP P^i_t & = & \prod_{(j,s)\in\ALP E^i_t\setminus(\ALP E^1_t\cap\ALP E^2_t)} \ALP Y^j_s\times \prod_{(j,s)\in\ALP F^i_t\setminus (\ALP F^1_t\cap \ALP F^2_t)} \ALP U^j_s.}
Note that $\ALP I^i_t$, $\ALP C_t$ and $\ALP P^i_t$ for $i=1,2$ and $t\in \{1,\ldots,T-1\}$ are finite dimensional Euclidean spaces. 

We let $\VEC I^i_t\in\ALP I^i_t$ denote the information available to Controller $i$ at time step $t\in\{1,\ldots,T-1\}$, which is a vector comprised of measurements and control actions that are observed by the controller. The common information of the controllers at a time step is defined as the vector of all random variables that are observed by both controllers at that time step. The private information of a controller at a time step is the vector of random variables that are not observed by the other controller. The common information is denoted by $\VEC C_t\in\ALP C_t$, and the private information of Controller $i$ is denoted by $\VEC P^i_t\in\ALP P^i_t$ at time step $t\in\{1,\ldots,T\}$.

A dynamic game is said to be one of symmetric information if $\VEC C_t=\VEC I^1_t = \VEC I^2_t$ at all time steps. We are interested in games where the controllers may have asymmetry in information, that is, $\VEC I^1_t \neq \VEC I^2_t$. An extreme example of a game of asymmetric information is the case when $\ALP P^1_t\neq \emptyset$ while $\ALP P^2_t=\emptyset$ for all time steps $t$. Another example of a game of asymmetric information is when the controllers recall their past information and share their observations after a delay of one time step, that is, $\ALP C_t = \ALP Y^{1:2}_{1:t-1}$ and $\ALP P^i_t =\ALP Y^i_t$ for all time steps $t$.

\subsection{Admissible Strategies of Controllers}

At every time step, Controller $i$ uses a {\it control law} $g^i_t: \ALP P^i_t\times \ALP C_t\rightarrow \ALP U^i_t$ to map its information to its action. We assume that the control law $g^i_t$ is a Borel measurable function, and denote the space of all such control laws by $\ALP G^i_t$.

A strategy of Controller $i$, which we define as the collection of its control laws over time, is denoted by $\mathbf g^i = (g^i_1,\ldots,g^i_{T-1})$ and the space of strategies of Controller $i$ is denoted by $\ALP G^i_{1:T-1}$. The pair of strategies of both controllers, $(\mathbf g^1,\mathbf g^2)\in\ALP G^1_{1:T-1}\times \ALP G^2_{1:T-1}$, is called the strategy profile of the controllers.

The total cost to Controller $i$, as a function of the strategy profile of the controllers, is
\beqq{J^i(\mathbf g^1,  \mathbf g^2) := \mathds{E}\Bigg[c^i_T(\VEC x_T) + \sum_{t=1}^{T-1}c^i_t(\VEC x_t,\VEC u^1_t,\VEC u^2_t)\Bigg],}
where $c^i_t$ is a non-negative continuous function of its arguments for $i\in\{1,2\}$ and $t\in\{1,\ldots,T-1\}$, and the expectation is taken with respect to the probability measure on the state and action processes induced by the choice of strategy profile $(\mathbf g^1,  \mathbf g^2)$.

A strategy profile $(\mathbf g^1,  \mathbf g^2)$ is said to be a Nash equilibrium of the game if it satisfies the following two inequalities
\beqq{J^1(\mathbf g^1,  \mathbf g^2) \leq J^1(\tilde{\mathbf g}^1,  \mathbf g^2),\quad\text{ and }\quad J^2(\mathbf g^1,  \mathbf g^2) \leq J^2(\mathbf g^1,  \tilde{\mathbf g}^2),} 
for all admissible strategies $\tilde{\mathbf g}^1\in\ALP G^1_{1:T-1}$ and $\tilde{\mathbf g}^2\in\ALP G^2_{1:T-1}$. 
% In this paper, we are interested in obtaining a subset of Nash equilibria that can be computed using a backward induction algorithm.

We assume that the state evolution equations, observation equations, the noise statistics, cost functions of the controllers and the information structures of the controllers are part of common knowledge. The game thus defined is referred to as game {\bf G1}.

\subsection{Assumption on Evolution of Information}
As noted above, each controller's information consists of common information and private information. We place the following condition on the evolution of common and private information of the controllers in game {\bf G1}.  
 \begin{assumption}\label{assm:lqginfoevolution}
The common and private information evolve over time as follows:
\begin{enumerate}
\item The common information increases with time, that is, $(\ALP E^1_t\cap \ALP E^2_t) \subset (\ALP E^1_{t+1}\cap \ALP E^2_{t+1})$ and $(\ALP F^1_t\cap \ALP F^2_t) \subset (\ALP F^1_{t+1}\cap \ALP F^2_{t+1})$ for all $t\in\{1,\ldots,T-1\}$. Let $\ALP D_{1,t+1}:= (\ALP E^1_{t+1}\cap \ALP E^2_{t+1})\setminus(\ALP E^1_{t}\cap \ALP E^2_{t})$ and $\ALP D_{2,t+1}:= (\ALP F^1_{t+1}\cap \ALP F^2_{t+1})\setminus(\ALP F^1_{t}\cap \ALP F^2_{t})$. Define
\beq{\label{eqn:commoninfo}\ALP Z_{t+1} = \prod_{(j,s)\in\ALP D_{1,t+1}} \ALP Y^j_s\times \prod_{(j,s)\in\ALP D_{2,t+1}} \ALP U^j_s.}
Then, $\VEC Z_{t+1}\in\ALP Z_{t+1} $ denotes the increment in common information from time $t$ to $t+1$ and we have 
\begin{equation}
\VEC Z_{t+1} = \zeta_{t+1}([\VEC P^{1\transpose}_t, \VEC P^{2\transpose}_t, \VEC U^{1\transpose}_t, \VEC U^{2\transpose}_t, \VEC Y^{1\transpose}_{t+1}, \VEC Y^{2\transpose}_{t+1}]^\transpose),\label{eq:lqgcommoninfo}
\end{equation}
where $\zeta_{t+1}$ is an appropriate projection function.
\item The private information evolves according to the equation
\begin{equation}
\VEC P^i_{t+1} = \xi^i_{t+1}([\VEC P^{i\transpose}_{t}, \VEC U^{i\transpose}_t,  \VEC Y^{i\transpose}_{t+1}]^\transpose). \label{eq:lqgprivateinfo}
\end{equation}
where $\xi^i_{t+1}$ is an appropriate projection function.
\end{enumerate}
\end{assumption}

% The first part of the assumption implies that the controllers do not forget their past common information, and therefore, the common information increases at every time step.

We now introduce a few notations in order to prove an important result. Let $\ALP S_t := \ALP X_t\times\ALP P^1_t\times\ALP P^2_t$ for $t\in\{1,\ldots,T-1\}$. Fix the strategy profile of the controllers as $(\VEC g^1,\VEC g^2)$. Let $\Pi_t$ be the conditional measure on the space of state and private informations, $\ALP S_t$,  given the common information $\VEC C_t$ at time step $t$. Thus,
\beqq{\Pi_t(d\VEC s_t) = \mathds{P}^{g^{1:2}_{1:t-1}}\left\{d\VEC s_t\Big|\VEC C_t\right\},}
where the superscript denotes the fact that the probability measure depends on the choice of control laws. The conditional probability measure $\Pi_t$ is a $\VEC C_t$-measurable random variable, whose realization, denoted by $\pi_t$, depends on the realization  $\VEC c_t$ of the common information. We now have the following result, which is a consequence of Assumption \ref{assm:lqginfoevolution}.

Let $\Gamma^i_t$ be a random measurable function from $\ALP P^i_t$ to $\ALP U^i_t$ defined as $\Gamma^i_t(\cdot):=g^i_t(\cdot,\VEC C_t)$ for $i=1,2$ and $t\in\{1,\ldots,T-1\}$, where the realization of $\Gamma^i_t$ is denoted by $\gamma^i_t$ and it depends on the realization of the random variable $\VEC C_t$. Thus, $g^i_t(\VEC P^i_t,\VEC c_t) = \gamma^i_t(\VEC P^i_t)$. We now have the following result about the evolution of conditional measure $\Pi_t$.
\begin{lemma}\label{lem:piupdate}
Fix the strategy profile $(\VEC g^1,\VEC g^2)\in\ALP G^1_{1:T-1}\times\ALP G^2_{1:T-1}$. If Assumption \ref{assm:lqginfoevolution} holds, then 
\beqq{\Pi_{t+1} = F_t(\Pi_t,\Gamma^1_t,\Gamma^2_t,\VEC Z_{t+1}),}
where $F_t$ is a fixed transformation which does not depend on the choice of control laws.
\end{lemma}
\begin{proof}
See Appendix \ref{app:piupdate}.
\end{proof}

% We now make an important assumption on the information structure of the controllers in the next subsection.
\subsection{Strategy Independence of Beliefs}
A crucial assumption, which forms the basis of the analysis in this paper, is the following.
\begin{assumption}[Strategy Independence of Beliefs] \label{assm:lqgseparation}
At any time $t$ and for any realization of common information $\VEC c_t$, the conditional probability measure $\pi_t$ on the state $\VEC X_t$ and the private information $(\VEC P^1_t, \VEC P^2_t)$ given the common information does not depend on the choice of control laws. In particular, if $\VEC z_{t+1}$ is a realization of the increment in the common information at time step $t+1$, then $\pi_{t+1}$ evolves according to the equation
\beq{\label{eqn:pit1}\pi_{t+1} = F_t(\pi_t, \VEC z_{t+1}),}
where $F_t$ is a fixed transformation that does not depend on the control laws.
\end{assumption}

Assumption \ref{assm:lqgseparation} allows us to define the conditional belief $\pi_t$ without specifying the control laws used. Another important consequence of Assumption \ref{assm:lqgseparation} is that these conditional beliefs on the state and private information admit Gaussian density functions. We make this precise in the following lemma.

\begin{lemma}\label{lem:pitgaussian}
For any time step $t\in \{1,\ldots,T\}$ and any realization of common information $\VEC c_t$, the common information based conditional measure $\pi_t$ admits a Gaussian density function.
\end{lemma}
\begin{proof}
See Appendix \ref{app:pitgaussian}.
\end{proof}

We henceforth call $\pi_t$ as {\it common information based conditional belief}. In the next subsection, we prove a result on the evolution of the mean and variance of the common information based conditional beliefs.

\subsection{Evolution of Conditional Beliefs}
Since the common information based conditional belief $\pi_t$ admits a Gaussian density at any time step $t$, $\pi_t$ is completely characterized by its mean ${\VEC m}_t$ and the covariance matrix $\Sigma_t$. The conditional covariance of a collection of jointly Gaussian random variables is data independent. 
% Since, by Assumption \ref{assm:lqgseparation}, the conditional distribution does not depend on the control laws, the covariance matrix at any time step is data independent for any choice of control laws. Therefore, 
Assumption \ref{assm:lqgseparation} allows us to derive the following result for game {\bf G1}.
\begin{lemma}\label{lemma:lqgmean}
The evolution of the conditional mean \[\VEC M_t := (\VEC M^0_t, \VEC M^1_t, \VEC M^2_t) =(\mathds{E}[\VEC X_t| \VEC C_t],\mathds{E}[\VEC P^1_t| \VEC C_t],\mathds{E}[\VEC P^2_t| \VEC C_t])\]  of the density function of common information based conditional belief is given as
\begin{equation}
{\VEC M}_{t+1} = F^1_t({\VEC M}_t,\VEC Z_{t+1}), \label{eq:lqgevolution1}
\end{equation}
where $F^1_t$ is a fixed affine transformation that does not depend on the strategies of the controllers. The evolution of conditional covariance matrix $\Sigma_t$ is given as
\begin{equation}
\Sigma_{t+1} = F^2_t(\Sigma_t), \label{eq:lqgevolution2}
\end{equation}
where $F^2_t$ is a fixed  transformation that does not depend on the strategies of the controllers.
\end{lemma}
\begin{proof}
See Appendix \ref{app:lqgmean}.
\end{proof}

% \begin{remark}
% Due to the result of Lemma \ref{lemma:lqgmean}, we henceforth write $f^i_t$ in \eqref{eqn:gitadmissible} in Assumption \ref{ass:restriction} as a function of the mean $\VEC m_t$ and $\VEC p^i_t$.
% \end{remark}

Examples of several classes of games that satisfy Assumptions \ref{assm:lqginfoevolution} and \ref{assm:lqgseparation} are given in \cite{nayyar2012a}. For example, if each controller acquires the realizations of the observations and the actions of the other controller with zero or one-step delay, then the corresponding game satisfies Assumptions \ref{assm:lqginfoevolution} and \ref{assm:lqgseparation}.

% In the next section, we define a new two-player symmetric information game, which we call game {\bf G2}, in which the players observe common information of the controllers in game {\bf G1}.

\section{Main Results}\label{sec:main}
Following the approach introduced in \cite{nayyar2012a}, we now construct a new game \textbf{G2} with two virtual players, where at every time step $t\in\{1,\ldots,T-1\}$, each virtual player observes the common information $\VEC C_t$, but not the private information of the controllers. Since the common information is nested (by Assumption \ref{assm:lqginfoevolution}), game {\bf G2} is a game of perfect recall. This game is intricately related to game {\bf G1}, and we exploit the symmetric information structure of game {\bf G2} to devise a computational scheme to compute a Nash equilibrium of game {\bf G1}. The steps taken to devise the scheme are as follows:
\begin{enumerate}
\item We formulate game {\bf G2} in the next three subsections. Further, we show that the common information based conditional mean $\VEC M_t$ is a Markov state of game {\bf G2} at time $t$.
% , which is controlled by the actions of the virtual players at that time step. Since the conditional mean at a time step is just a function of the common information at that time step, game {\bf G2}, with the state $\VEC M_t$, is a game of perfect information.
\item In Subsection \ref{sub:relation}, we show that any Nash equilibrium of game {\bf G2} can be used to obtain a Nash equilibrium of game {\bf G1}, and vice-versa.
\item We focus on Markov perfect equilibria of game {\bf G2} and provide a backward induction characterization of such equilibria. An equilibrium of game {\bf G1} obtained from a Markov perfect equilibrium of game {\bf G2} is called {\it common information based Markov perfect equilibrium} of game {\bf G1}.
\item We interpret the backward induction characterization of common information based Markov perfect equilibrium in terms of a sequence of one-stage Bayesian games.
\end{enumerate}

We now turn our attention to formulating game {\bf G2}. At time step $t$ and for each realization $\VEC c_t$ of the common information, virtual player $i$ selects a measurable function $\gamma^i_t:\ALP P^i_t\rightarrow\ALP U^i_t$. The action space of virtual player $i$ at time $t$ is denoted by $\ALP A^i_t$, and it is defined as
\beq{\label{eqn:alpAit}\ALP A^i_t :=\{\gamma^i_t:\ALP P^i_t\rightarrow\ALP U^i_t \text{ such that }\gamma^i_t \text{ is a Borel measurable map}\}.}
% We endow the space $\ALP A^i_t$ with $L_2$ norm, which makes it a topological space.
% Note that $\ALP A^i_t$ is a finite dimensional Euclidean space.

We call the actions taken by virtual players as ``prescriptions'' due to the following reason: After observing the common information $\VEC c_t$ at time step $t$, virtual player $i\in\{1,2\}$ computes the equilibrium prescription $\gamma^i_t\in\ALP A^i_t$, and prescribes it to Controller $i$. The controllers evaluate the prescriptions based on the realizations of their private informations, to compute their actions at that time step.

% Now, we identify a Markov state, strategy spaces and cost functions of the virtual players in game {\bf G2} in the next few subsections. Then, we discuss the relation between the two games {\bf G1} and {\bf G2}. 

\subsection{Admissible Strategies of Virtual Players in Game {\bf G2}}
A map $\chi^i_t:\ALP C_t\rightarrow \ALP A^i_t$ denotes the control law of virtual player $i\in\{1,2\}$ at time step $t\in\{1,\ldots,T-1\}$. The control law $\chi^i_t$ maps common information at time $t$ to a prescription, which itself maps private information of Controller $i$ at time $t$ to the control action of Controller $i$. Thus, a choice of $\chi^i_t$ induces a map from $\ALP C_t\times\ALP P^i_t$ to $\ALP U^i_t$, which we denote by $\chi^i_t(\cdot)(\cdot)$. We say $\chi^i_t$ is admissible if
% \beqq{\int_{\ALP P^i_t} \|\chi^i_t(\VEC c_t)(\VEC p^i_t)\|^2\pr{d\VEC p^i_t|\VEC c_t}<\infty \quad   \text{ for every } \VEC c_t\in\ALP C_t, \\
% \text{ and }\quad
\beqq{\chi^i_t(\cdot)(\cdot) \quad  \text{ is a Borel measurable function from } \ALP C_t\times \ALP P^i_t \text{ to } \ALP U^i_t. }
The set of all such admissible control laws is denoted by $\ALP H^i_t$, $i\in\{1,2\}$ and $t\in\{1,\ldots,T-1\}$. The collection of control laws at all time steps of virtual player $i$ is called the strategy of that virtual player, and it is denoted by $\chi^i :=\{\chi^i_1,\ldots,\chi^i_{T-1}\}$. The space of all strategies of the virtual player $i$, denoted by $\ALP H^i_{1:T-1}$, is called the strategy space of that virtual player. A strategy tuple $(\chi^1,\chi^2)$ is called the strategy profile of virtual players. 

\begin{definition}\label{def:var}
For $i\in\{1,2\}$ and $t\in\{1,\ldots,T-1\}$, let $\varrho^i_t:\ALP H^i_t\rightarrow\ALP G^i_t$ be an operator that takes a function $\chi^i_t:\ALP C_t\rightarrow \ALP A^i_t$ as its input and returns a measurable function $g^i_t:\ALP P^i_t\times \ALP C_t\rightarrow\ALP U^i_t$ as its output, that is $ g^i_t = \varrho^i_t(\chi^i_t)$, such that $g^i_t(\VEC p^i_t,\VEC c_t) := \chi^i_t(\VEC c_t)(\VEC p^i_t)$ for all $\VEC c_t\in\ALP C_t$ and $\VEC p^i_t\in\ALP P^i_t$. For a collection of functions $\chi^i :=\{\chi^i_1,\ldots,\chi^i_{T-1}\}$, let $\varrho^i(\chi^i)$ be defined as the set $\{\varrho^i_1(\chi^i_1),\ldots,\varrho^i_{T-1}(\chi^i_{T-1})\}$.

Similarly, we let $\varsigma^i_t:\ALP G^i_t\rightarrow\ALP H^i_t$ be the operator such that $\varsigma^i_t\circ \varrho^i_t = \text{id}_{\ALP H^i_t}$ and $\varrho^i_t\circ \varsigma^i_t = \text{id}_{\ALP G^i_t}$. Thus, for $g^i_t\in\ALP G^i_t$, if $\chi^i_t=\varsigma^i_t(g^i_t)$, then $\chi^i_t(\VEC c_t)(\VEC p^i_t) := g^i_t(\VEC p^i_t,\VEC c_t)$ for all $\VEC c_t\in\ALP C_t$ and $\VEC p^i_t\in\ALP P^i_t$. Similar to the expression above, for a collection of functions $\VEC g^i :=\{g^i_1,\ldots,g^i_{T-1}\}$, let $\varsigma^i(\VEC g^i)$ be defined as the set $\{\varsigma^i_1(g^i_1),\ldots,\varsigma^i_{T-1}(g^i_{T-1})\}$.{\hfill $\Box$}
\end{definition}
% 
% Henceforth, we let $\gamma^i_t\in\ALP A^i_t$ denote the prescription chosen by the virtual player $i$ at time step $t$ given that the common information at that time step is $\VEC c_t$. This means, for every realization of common information $\VEC c_t$, $\gamma^i_t:=\chi^i_t(\VEC c_t)$ is a realization of the prescription chosen by the virtual player $i$ at time step $t$ and it lies in the set $\ALP A^i_t$.

\subsection{Cost Functions for Virtual Players}
The cost functions of the virtual players are defined as follows: Fix a time step $t\in\{1,\ldots,T-1\}$ and a virtual player $i$. Let $\pi$ denote a normal distribution on the space $\ALP S_t = \ALP X_t\times\ALP P^1_t\times\ALP P^2_t$ with mean $\VEC m$ and variance $\Sigma_t$, where $\Sigma_t$ is given by the result in Lemma \ref{lemma:lqgmean}, and let $(\gamma^1,\gamma^2)$ be a prescription pair chosen by the virtual players. Define the cost function $\tilde{c}^i_t:\ALP S_t\times\ALP A^1_t\times\ALP A^2_t\rightarrow\Re_+$ of virtual player $i$ at that time step $t\in\{1,\ldots,T-1\}$ to be 
\beqq{\tilde{c}^i_t(\VEC m,\gamma^1,\gamma^2) = \int_{\ALP S_t} c^i_t(\VEC x_t,\gamma^1(\VEC p^1_t),\gamma^2(\VEC p^2_t))\pi(d\VEC s_t),}
where one can view $\VEC x_t,\VEC p^1_t$ and $\VEC p^2_t$ as appropriate projections of the variable $\VEC s_t$. We now define the cost function of the virtual players at the final time step. Let $\pi$ be a Gaussian distribution with mean $\VEC m$ and variance $\Sigma_T$. The cost functions of the virtual players at the final time step is $\tilde{c}^i_T(\VEC m) = \int_{\ALP X_T} c^i_T(\VEC x_T)\pi(d\VEC x_T)$. The total cost for virtual player $i$ is given by
\beqq{\tilde J^i(\chi^1,\chi^2) = \ex{\tilde{c}^i_T(\VEC M_T)+\sum_{t=1}^{T-1}\tilde{c}^i_t(\VEC M_t,\Gamma^1_t,\Gamma^2_t)},} 
where the expectation is taken with respect to the probability measure induced on the mean of common information based conditional belief by the choice of strategies $(\chi^1,\chi^2)$. We have the following claim about the expected cost of game {\bf G2} given a pair of strategy profiles of controllers in game {\bf G1} and {\it vice versa}.

\begin{lemma}\label{lem:vpcost}
Let $(\VEC g^1,\VEC g^2)\in\ALP G^1_{1:T-1}\times\ALP G^2_{1:T-1}$, and let $(\chi^1,\chi^2)$ be defined as $\chi^i:=\varsigma^i(\VEC g^i), i=1,2$. Then, $J^i(\VEC g^1,\VEC g^2) = \tilde J^i(\chi^1,\chi^2)$ for $i=1,2$.

Conversely, let $(\chi^1,\chi^2)\in\ALP H^1_{1:T-1}\times\ALP H^2_{1:T-1}$, and let $(\VEC g^1,\VEC g^2)$ be defined as $\VEC g^i:=\varrho^i(\chi^i), i=1,2$. Then, $\tilde J^i(\chi^1,\chi^2)=J^i(\VEC g^1,\VEC g^2)$ for $i=1,2$.
\end{lemma}
\begin{proof}
See Appendix \ref{app:vpcost}.
\end{proof}

\subsection{A Markov State of Game {\bf G2}}
Recall from Lemma \ref{lemma:lqgmean} that given a realization $\VEC c_t$ of common information, the common information based conditional belief $\mathds{P}\{d\VEC s_t|\VEC c_t\}$ admits a Gaussian density function with mean $\VEC m_t$ and variance $\Sigma_t$. Our next result is that the mean $\VEC M_t$ is a controlled Markov chain, and it is controlled by the actions taken (that is, the prescriptions chosen) by the virtual players.

\begin{lemma}\label{lemma:lqgmarkov}
The process $\{\VEC M_t\}_{t\in\{1,\ldots,T\}}$ is a controlled Markov process with the virtual players' prescriptions as the controlling actions. In particular, conditioned on the realization  $\VEC m_t$ of  $\VEC M_t$ and the prescriptions $(\gamma^1_t, \gamma^2_t)$ of the virtual players, the conditional mean $\VEC M_{t+1}$ at the next time step is independent of the current common information, the past conditional means and past prescriptions. Equivalently, this fact is expressed as
\beqq{\mathds{P}\{\VEC M_{t+1}\in \SF M_{t+1}|\VEC c_t, \VEC m_{1:t},\gamma^{1:2}_{1:t}\} = \mathds{P}\{\VEC M_{t+1}\in \SF M_{t+1}|\VEC m_{t},\gamma^{1:2}_t\}}
for all Borel sets $\SF M_{t+1}\subset \ALP S_{t+1}$.
\end{lemma}
\begin{proof}
See Appendix \ref{app:lqgmarkov}.
\end{proof}

It should be noted that the update equation of the Markov process $\VEC M_t$ is induced by the state dynamics, observation equations, and information structure of the controllers of game {\bf G1}. The main differences between the structures of the two games {\bf G1} and {\bf G2} are summarized in the table below.\\

\begin{center}
\renewcommand{\arraystretch}{1.5}
\begin{tabular}{|l|c|c|}
\hline
At time step $t\in\{1,\ldots,T-1\}$ & Game {\bf G1} & Game {\bf G2}\\\hline
State of the game& $\VEC X_t\in\ALP X_t$ & $\VEC M_t\in\ALP S_t$\\\hline
Action of Player $i$ & $\VEC U^i_t\in\ALP U^i_t$ & $\gamma^i_t\in\ALP A^i_t$\\\hline
Information of Player $i$ & $\VEC I^i_t\in\ALP I^i_t$ & $\VEC C_t\in\ALP C_t$\\\hline
Cost function of Player $i$ & $c^i_t$ & $\tilde c^i_t$\\
\hline
\end{tabular}
\end{center}
\vspace{10pt}

Since both virtual players observe the same information (the common information between the controllers) and the common information always increases by Assumption \ref{assm:lqginfoevolution}, game {\bf G2} between virtual players is a game of symmetric information with perfect recall. 

\subsection{Relation between Games {\bf G1} and {\bf G2}}\label{sub:relation}
In the next theorem, we show that any Nash equilibrium of game {\bf G1} can be used to compute a Nash equilibrium for game {\bf G2} and vice versa.
\begin{theorem}\label{thm:lqgequiv}
Let $(\chi^{1\star}, \chi^{2\star})$ be a Nash equilibrium strategy profile of game \textbf{G2}. Then, the strategy profile $(\mathbf g^{1\star}, \mathbf g^{2\star})$ for game \textbf{G1}, defined as $\VEC g^{i\star} =\varrho^i(\chi^{i\star}), i=1,2$, forms a Nash equilibrium strategy profile of game \textbf{G1}. Conversely, if $(\mathbf g^{1\star}, \mathbf g^{2\star})$ is a Nash equilibrium strategy profile of game {\bf G1}, then the strategy profile $(\chi^{1\star}, \chi^{2\star})$, defined by $\chi^{i\star}:=\varsigma^i(\VEC g^{i\star}), i=1,2,$ is a Nash equilibrium strategy profile for game {\bf G2}.
\end{theorem}
\begin{proof}
See Appendix \ref{app:lqgequiv}.
\end{proof}

In light of the theorem above, we want to compute a Nash equilibrium of game {\bf G2}, and then project the solution back to the original game using the operators $\varrho^1$ and $\varrho^2$ as introduced in Definition \ref{def:var}.

Since at any time step $t$, both virtual players observe the common information $\VEC C_t$, the virtual players can compute the mean $\VEC M_t$ of the common information based conditional belief. Since the mean $\VEC M_t$ is a Markov state of game {\bf G2} and both virtual players know its realization, game {\bf G2} is a game of perfect state information. Also note that in game {\bf G2}, the cost functions of the virtual players are stagewise-additive. A natural solution concept to compute the Nash equilibrium of a game of perfect information with stagewise-additive cost function is Markov perfect equilibrium \cite{fudenberg1991}. We define the Markov perfect equilibrium of {\bf G2} in the next subsection and prove the main result of the section.

\subsection{Markov Perfect Equilibrium of Game {\bf G2}}
Fix virtual player $i$'s control laws $\chi^i_{1:T-1}$ such that the prescription at time step $t$ is only a function of $\VEC M_t$, say $\chi^i_t(\VEC C_t) = \psi^i_t(\VEC M_t)$ for some function $\psi^i_t:\ALP S_t\rightarrow\ALP A^i_t$, such that $\psi^i_t(\cdot)(\cdot)$ is a measurable function from $\ALP S_t\times\ALP P^i_t$ to $\ALP U^i_t$ at all time step $t\in\{1,\ldots,T-1\}$. Let us use $\bar{\ALP H}^i_t$ to denote the set of all such maps $\psi^i_t$, and note that $\bar{\ALP H}^i_t\subset\ALP H^i_t$. Then, virtual player $j$'s, $j\neq i$, (one-person) optimization problem is to minimize its stagewise additive cost functional that depends on $\VEC M_t$ and on virtual player $i$'s fixed strategy. Thus, virtual player $j$ needs to solve a finite horizon Markov decision problem with state space $\ALP S_t$ and action space $\ALP A^{j}_t$ at time step $t$. This is made precise in our next result.

\begin{lemma}\label{lemma:lqginfostatelemma}
Consider game {\bf G2} among virtual players. Assume that virtual player $i$ is using the strategy $\{\psi^i_1,\ldots,\psi^i_{T-1}\}\in\bar{\ALP H}^i_{1:T-1}$, that is, virtual player $i$ selects the prescriptions at time step $t$ only as a function of the mean $\VEC M_t$ of the common information based conditional belief $\Pi_t$: 
\[ \Gamma^i_t = \psi^i_t(\VEC M_t),\qquad t\in\{1,\ldots, T-1\}.\]
Then, for the fixed strategy of virtual player $i$, virtual player $j$'s ($j\neq i, j\in\{1,2\}$) one-sided optimization problem is a finite horizon Markov decision problem with state $\VEC M_t$, control action $\gamma^j_t$, and cost as $\tilde{c}^j_t(\VEC M_t,\gamma^j_t,\psi^i_t(\VEC M_t))$ at time step $t\in\{1,\ldots,T-1\}$ and terminal cost $\tilde{c}^j_T(\VEC M_T)$.
\end{lemma}
\begin{proof}
See Appendix \ref{app:lqginfostatelemma}.
\end{proof}

Note that game {\bf G2} is a dynamic game of perfect information and perfect recall. A Markov strategy of a virtual player is defined as a collection of control laws of that virtual player at all time steps such that the control law at time step $t$ is a measurable map of common information based conditional mean (state of game {\bf G2}) to its action space at that time step. Lemma \ref{lemma:lqginfostatelemma} states that if one virtual player sticks to Markov strategy, then the other virtual player's one-sided optimization problem is a finite horizon Markov decision problem. Under certain assumptions on the cost functions of the virtual players\footnote{See, for example, \cite[Section 3.3]{lerma1996} for a set of such assumptions.}, there exists a Markov strategy of the other virtual player that achieves the minimum in its Markov decision problem. Thus, there is no incentive for the other virtual player to search for optimal strategies outside the class of Markov strategies. This is an important 
observation for game {\bf G2}, because one can define a refinement concept for Nash equilibrium, called Markov perfect equilibrium \cite{fudenberg1991}, for game {\bf G2}.

\begin{definition}
A strategy profile $(\psi^{1\star}_{1:T-1},\psi^{2\star}_{1:T-1})\in \ALP H^1_{1:T-1}\times\ALP H^2_{1:T-1}$ is said to be a Markov perfect equilibrium \cite{fudenberg1991} of game \textbf{G2} if (i) at each time $t$, the control laws of the virtual players at time step $t$ are functions of the mean of the common information based conditional belief $\VEC M_t$, that is, $\psi^i_t\in\bar{\ALP H}^i_t$, and (ii) for all time steps $t\in\{1,\ldots,T-1\}$, the strategy profiles $(\psi^{1\star}_{t:T-1},\psi^{2\star}_{t:T-1})$ form a Nash equilibrium for the sub-game starting at time step $t$ of game {\bf G2}.
\end{definition}

It should be noted that Markov perfect equilibrium is a refinement concept for Nash equilibria of games in which players make perfect state observations. In game {\bf G2} among virtual players, a strategy profile that is {\it not} a Markov perfect equilibrium either depends on the common information (and not just on the mean $\VEC M_t$), or is not a Nash equilibrium of every sub-game in game {\bf G2}, or both. We reemphasize this point later in Section \ref{sec:information}.  

Given a Markov perfect equilibrium of \textbf{G2}, we can construct a corresponding Nash equilibrium of game \textbf{G1} using Theorem~\ref{thm:lqgequiv}. We refer to the class of Nash equilibria of \textbf{G1} that can be constructed from the Markov perfect equilibria of \textbf{G2} as the \emph{common information based Markov perfect equilibria} of game \textbf{G1}.
\begin{definition}
If $(\psi^{1\star}_{1:T-1},\psi^{2\star}_{1:T-1})$ is a Markov perfect equilibrium of game \textbf{G2}, then the strategy profile $(\VEC g^{1\star}, \VEC g^{2\star})$ of the form $\VEC g^{i\star} = \varrho^i(\psi^{i\star}_{1:T-1}), i=1,2,$  is called \emph{common information based Markov perfect equilibrium} of game \textbf{G1}.{\hfill$\Box$}
\end{definition}

A similar concept was introduced for finite games with asymmetric information in our earlier work \cite{nayyar2012a}.

\subsection{Computation of Markov Perfect Equilibrium of Game {\bf G2}}\label{sub:MPEG2}
In this subsection, we characterize Markov perfect equilibrium of game \textbf{G2} using value functions that depend on the mean of the common information based conditional belief. 

% It is worth noting that there may be other Nash equilibria in game {\bf G2} among virtual players, but our focus is mainly on a subclass of Nash equilibria that can be computed using a backward induction algorithm and depends only on the state.

\begin{theorem}\label{thm:mpevirtual}
Consider a strategy pair $(\psi^{1\star}_{1:T-1},\psi^{2\star}_{1:T-1})\in\bar{\ALP H}^1_{1:T-1}\times\bar{\ALP H}^2_{1:T-1}$. Define functions $V^i_t:\ALP X_t\times\ALP P^1_{t}\times\ALP P^2_{t}\rightarrow\Re$, called expected value functions of Controller $i$ at time $t$, as follows:
\begin{enumerate}
\item For each possible realization $\VEC m = (\VEC m^0,\VEC m^1,\VEC m^2)$ of $\VEC M_{T}$, define the value functions:
\beq{\label{eqn:viT} V^i_{T}(\VEC m) := \tilde c^i_T(\VEC m) = \mathds{E}[c^i_T(\VEC X_T)|\VEC M_T = \VEC m]\qquad i\in\{1,2\}.}
% where $\Sigma^0_T$ is the upper left block of $\Sigma_T$ corresponding to the conditional covariance matrix of $\VEC X_T$ given common information $\VEC C_T$.
\item For $t=T-1,\ldots,1$, and for each possible realization $\VEC m$ of $\VEC M_{t}$, define  the value functions:
\beq{V^i_{t}(\VEC m) &:=& \min_{\tilde\gamma^i_t\in\ALP A^i_t} \mathds{E}\Big[ \tilde c^i_t(\VEC M_t,\gamma^1_t,\gamma^2_t) +V^i_{t+1}(F^1_t(\VEC M_t, \VEC Z_{t+1})) \nonumber\\
& & \qquad\qquad\Big|\VEC M_{t}=\VEC m, \gamma^i_t =\tilde\gamma^i_t , \gamma^{-i}_t = \psi^{-i\star}_{t}(\VEC m)\Big]\qquad i\in\{1,2\},\qquad\label{eqn:vit}}
assuming that the minimum exists in the equation above. Then, a necessary and sufficient condition for $(\psi^{1\star}_{1:T-1},\psi^{2\star}_{1:T-1})$ to be a Markov perfect equilibrium of \textbf{G2} is that for every time step $t\in\{1,\ldots,T-1\}$, $i\in\{1,2\}$ and for every realization $\VEC m$ of $\VEC M_t$,
\beq{\label{eqn:psiitm}\psi^{i\star}_t(\VEC m) \in \underset{\tilde\gamma^i_t\in \ALP A^i_t}{\arg\min} \;\;\mathds{E}\Big[ \tilde c^i_t(\VEC m,\tilde\gamma^i_t, \psi^{-i\star}_{t}(\VEC m)) +V^i_{t+1}(F^1_t(\VEC m, \VEC Z_{t+1}))\Big].}
\end{enumerate}
\end{theorem}
\begin{proof}
See Appendix \ref{app:mpevirtual}.
\end{proof}

To show that the sub-game admits a Nash equilibrium at time step $t$ requires a fixed-point argument, which means that the reaction curves of the virtual players intersect in the product of their strategy spaces $\bar{\ALP H}^1_t\times\bar{\ALP H}^2_t$ \cite{Basarbook}. In the next section, we show that if the cost functions of the players are quadratic functions of their arguments, then under certain conditions on the cost functions, the reaction curves of the virtual players intersect at a unique point. Thus, under those assumptions, a unique common information based Markov perfect equilibrium exists in LQG games with asymmetric information.

\begin{remark}\label{rem:multiequilibria}
As stated earlier, Markov perfect equilibrium is only a subclass of Nash equilibria of game {\bf G2}. Game {\bf G2} (and equivalently, the corresponding game {\bf G1}) may have several other Nash equilibria besides Markov perfect equilibrium. An example of a two-player two-stage game of asymmetric information in which there is a continuum of Nash equilibria is presented in Section \ref{sec:information}.{\hfill $\Box$}
%  In that example, a family of Nash equilibria of the game among virtual players is computed in which at time step $t=2$, the family of equilibria is computed using the common information $\VEC C_2$ in addition to $\VEC M_2$.
\end{remark}

\subsection{One-stage Bayesian Games}\label{sub:stagebayesian}
At any time step $t$, let $\VEC m_t$ and $\VEC x_t$, respectively, be realizations of the common information based conditional mean and the state. Let us rewrite the expressions of the expected cost-to-go functions in \eqref{eqn:vit} as
\beq{\label{eqn:hatcit}\bar{c}^i_t(\VEC m_t;\VEC x_t,\VEC u^1_t,\VEC u^2_t):=c^i_t(\VEC x_t,\VEC u^1_t,\VEC u^2_t) +\ex{V^i_{t+1}(F^1_t(\VEC m_t, \VEC Z_{t+1}))\Big| \VEC x_t,\VEC u^1_t,\VEC u^2_t}.}
These are the cost-to-go functions for the two controllers in game {\bf G1} if both controllers stick to the common information based Markov perfect equilibrium for all time step $s>t$. For a realization $\VEC m_t$, Controller $i$ chooses a map $\gamma^i_t:\ALP P^i_t\rightarrow\ALP U^i_t$. We assume that the probability measure on $\ALP X_t\times\ALP P^1_t\times\ALP P^2_t$ admits a Gaussian density function with mean $\VEC m_t$ and variance $\Sigma_t$. One can notice that this is precisely the setup of a Bayesian game. Therefore, we say that the game between Controllers $1$ and $2$, with cost functions $\bar{c}^1_t(\VEC m_t;\cdot)$ and $\bar{c}^2_t(\VEC m_t;\cdot)$, respectively, is the {\it one-stage Bayesian game at time step $t$ with mean $\VEC m_t$}. 

% If $(\gamma^{1\star}_t,\gamma^{2\star}_t)\in\ALP A^1_t\times\ALP A^2_t$ is a Nash equilibrium of the one-stage Bayesian game, then $\psi^{i\star}_t(\VEC m_t) = \gamma^{i\star}_t$ is the Markov perfect equilibrium strategy for virtual player $i$ at time $t$ evaluated at $\VEC m_t$ in game {\bf G2}.

\textbf{An algorithm to compute common information based Markov perfect equilibrium:}

We can now describe a backward induction process to find a common information based Markov perfect equilibrium of game \textbf{G1} using a sequence of one-stage Bayesian games. We proceed as follows: \\
\underline{\textbf{Algorithm 1:}}
\begin{enumerate}
\item At the terminal time $T-1$, for each realization $\VEC m$ of the common information based conditional mean at time $T-1$, we define a one-stage Bayesian game  $SG_{T-1}(\VEC m)$ where 
\begin{enumerate}
\item The probability distribution on $(\VEC X_{T-1}, \VEC P^1_{T-1}, \VEC P^2_{T-1})$, denoted by $\pi$, is a Gaussian distribution with mean $\VEC m$ and covariance $\Sigma_{T-1}$.
\item  Agent\footnote{Agent $i$ can be thought to be the same as Controller $i$. We use a different name here in order to maintain the distinction between games \textbf{G1} and $SG_{T-1}(\VEC m)$.} $i$  observes $\VEC P^i_{T-1}$ and chooses action $\VEC U^i_{T-1}$, $i=1,2$. 
% \item  Agent $i$'s cost is $c^i_{T-1}(\VEC X_{T-1},\VEC U^1_{T-1},\VEC U^2_{T-1}) + c^i_T(\VEC X_T) $, $i=1,2$.
\item  Agent $i$'s cost is $\bar{c}^i_{T-1}(\VEC m;\VEC X_{T-1},\VEC U^1_{T-1},\VEC U^2_{T-1})$, $i=1,2$.
\end{enumerate}
A Bayesian Nash equilibrium\footnote{See \cite{fudenberg1991, Osborne, Myerson_gametheory} for a definition of Bayesian Nash equilibrium.} of this game is a pair of strategies $(\gamma^{1*},\gamma^{2*})$, where $\gamma^{i*}:\ALP P^i_{T-1}\rightarrow\ALP U^i_{T-1}$ is a measurable function such that for any realization $\VEC p^i\in\ALP P^i_{T-1}$, $\gamma^{i*}(\VEC p^i)$ is a solution of the minimization problem 
% \[ \min_{\VEC u^i} \mathds{E}^{\pi}[c^i_{T-1}(\VEC X_{T-1},\VEC u^i, \gamma^{j*}(\VEC P^{j}_{T-1})) + c^i_T(\VEC X_T)|\VEC P^i_{T-1} = \VEC p^i ], \]
\[ \min_{\VEC u^i} \mathds{E}^{\pi}[\bar{c}^i_{T-1}(\VEC m;\VEC X_{T-1},\VEC u^i, \gamma^{j*}(\VEC P^{j}_{T-1}))|\VEC P^i_{T-1} = \VEC p^i ], \]
where $j \neq i$ and the superscript $\pi$ denotes that the expectation is with respect to the distribution $\pi$. If a Bayesian Nash equilibrium $(\gamma^{1*},\gamma^{2*})$ of $SG_{T-1}(\VEC m)$ exists, denote the corresponding expected equilibrium costs as $V^i_{T-1}(\VEC m), i=1,2$, and define $\psi^i_T(\VEC m) := \gamma^{i*}$, $i=1,2$.

\item At time $t < T-1$, for each realization $\VEC m$ of the common information based conditional mean at time $t$, we define the one-stage Bayesian   game $SG_t(\VEC m)$ where
\begin{enumerate}
\item The probability distribution on $(\VEC X_t, \VEC P^1_t, \VEC P^2_t)$, denoted by $\pi$, admits a Gaussian density function with mean $\VEC m$ and covariance $\Sigma_{t}$.
\item  Agent $i$ observes $\VEC P^i_t$ and chooses action $\VEC U^i_t$, $i=1,2$. 
% \item  Agent $i$'s cost is $c^i(\VEC X_t,\VEC U^1_t,\VEC U^2_t) + V^i_{t+1}(F^1_t(\VEC m, \VEC Z_{t+1}))$, $i=1,2$. 
\item  Agent $i$'s cost is $\bar{c}^i_t(\VEC m;\VEC X_t,\VEC U^1_t,\VEC U^2_t)$, $i=1,2$. 
\end{enumerate}
% Recall that the common information based conditional mean for the next time step is $\VEC M_{t+1} = F^1_t(\VEC m, \VEC Z_{t+1})$ and $\VEC Z_{t+1}$ is given by \eqref{eqn:commoninfo}.
A Bayesian Nash equilibrium of this game is a pair of strategies $(\gamma^{1*},\gamma^{2*})$, where $\gamma^{i*}:\ALP P^i_{T-1}\rightarrow\ALP U^i_{T-1}$ is a measurable function such that for any realization $\VEC p^i\in\ALP P^i_t$, $\gamma^{i*}(\VEC p^i)$ is a solution of the minimization problem 
% \[ \min_{\VEC u^i} \mathds{E}^{\pi}[c^i_t(\VEC X_t,\VEC u^i, \gamma^{j}(\VEC P^{j}_t)) +V^i_{t+1}(F^1_t(\VEC m, \VEC Z_{t+1})) |\VEC P^i_t = \VEC p^i ], \]
\[ \min_{\VEC u^i} \mathds{E}^{\pi}[\bar{c}^i_t(\VEC m;\VEC X_t,\VEC u^i, \gamma^{j}(\VEC P^{j}_t)) |\VEC P^i_t = \VEC p^i ], \]
where $j \neq i, i,j=1,2,$ when control actions $\VEC U^i_t =\VEC u^i$ and $\VEC U^j_t = \gamma^{j}(\VEC P^{j}_t)$ are used. The expectation is taken with respect to the Gaussian distribution with mean $\VEC m$ and covariance $\Sigma_t$. If a Bayesian Nash equilibrium $(\gamma^{1*},\gamma^{2*})$ of $SG_t(\VEC m)$ exists, denote the corresponding expected equilibrium costs as $V^i_t(\VEC m), i=1,2$ and define $\psi^i_t(\VEC m) := \gamma^{i*}$, $i=1,2$.
%Consider any Bayesian Nash equilibrium $\gamma^{1*},\gamma^{2*}$ of $SG_t(\pi)$ and denote the equilibrium costs as $V^i_t(\pi)$. Define $\psi^i_t(\pi) := \gamma^{i*}$, $i=1,2$. 
\end{enumerate}

\begin{theorem}\label{thm:backward_ind}
The strategies $\boldsymbol \psi^i = (\psi^i_1,\psi^i_2,\ldots,\psi^i_{T-1})$, $i=1,2,$ defined by the backward induction process described in Algorithm 1 form a Markov perfect equilibrium of game \textbf{G2}. Consequently, strategies $\VEC g^1$ and $ \VEC g^2$ defined as
\[ g^i_t(\cdot, \VEC c_t) := \psi^i_t(\VEC m_t), \quad i\in\{1,2\},\:t\in\{1,\ldots,T-1\} \]
form a common information based Markov perfect equilibrium of game \textbf{G1}.
\end{theorem}
\begin{proof}
To prove the result, we just need to observe that the strategies defined by the backward induction procedure of Algorithm 1  satisfy the  conditions of Theorem \ref{thm:mpevirtual} and hence form a Markov perfect equilibrium of game \textbf{G2}.
\end{proof}

In the next section, we consider LQG games, in which the cost functions of the controllers at any time step are quadratic in the state and actions of the controllers. Under certain sufficient conditions on the cost functions of the controllers, we prove that the one-stage Bayesian game at any time $t\in\{1,\ldots,T-1\}$ with any mean $\VEC m_t\in\ALP S_t$ admits a unique Bayesian Nash equilibrium. We follow the steps of the algorithm above to prove that every LQG game with cost functions satisfying certain conditions admits a unique common information based Markov perfect equilibrium.

\section{Game with Quadratic Cost Functions}\label{sec:lqggames}
Let us now consider the special class of games where the stagewise cost functions $c^i_T$ and $c^i_t$ are quadratic functions of their arguments:
\beqq{c^i_T(\VEC X_{T}) = \VEC X_{T}^\transpose R^i_{11} \VEC X_{T}, & & c^i_t(\VEC X_t,\VEC U^1_t,\VEC U^2_t) = \bmat{\VEC X_t^\transpose , \VEC U^{1\transpose}_t, \VEC U^{2\transpose}_t} R^i \bmat{\VEC X_t \\ \VEC U^1_t\\ \VEC U^2_t}  \\
\quad\text{ where } R^i&:=&\bmat{R^i_{11} & R^i_{12} & R^i_{13}\\R_{12}^{i\transpose} & R^i_{22} & R^i_{23}\\R_{13}^{i\transpose} & R_{23}^{i\transpose} & R^i_{33}},}
$R^i_{11}\geq 0$ and $R^i_{ii}>0$ for $i\in\{1,2\}$. We refer to Gaussian games in which the cost functions are of the form above as dynamic LQG games.

Before we analyze the dynamic LQG game, we first formulate and compute the Nash equilibrium of a static auxiliary (Bayesian) game in the next subsection. We use the result of this auxiliary game to compute the Nash equilibrium strategies of the dynamic game. One of the main results of this section is that any LQG game that satisfies a certain assumption on the cost functions in addition to Assumptions \ref{assm:lqginfoevolution} and \ref{assm:lqgseparation} admits a unique common information based Markov perfect equilibrium in the class of all Borel measurable strategy profiles of the controllers. We prove this in two steps:
\begin{enumerate}
\item The first step consists of computing a Bayesian Nash equilibrium of a particular two-player static game with asymmetric information. This is done in Subsection \ref{sec:basargame}.
\item We then exploit the uniqueness of Nash equilibrium, the structure of the Nash equilibrium strategies of the controllers, and the expected equilibrium costs to the controllers to obtain the main result for LQG games in Subsection \ref{sub:generallqg}. 
\end{enumerate}

\subsection{An Auxiliary Game, {\bf AG1}}\label{sec:basargame}

The static Bayesian game is described as follows: $\VSF X, \VSF Y^1, \VSF Y^2$ are jointly Gaussian random vectors such that $\VSF Y^i = H^i \VSF X$ for some matrix $H^i$ of appropriate dimensions. The mean and covariance of the three-tuple $\VSF X, \VSF Y^1, \VSF Y^2$ are given by
\[ \VSF m = \bmat{ \VSF m_{x}  \\ \VSF m_{y^1} \\\VSF m_{y^2}} \quad \Sigma =\bmat{\Sigma_{xx} & \Sigma_{xy^1} & \Sigma_{xy^2} \\\Sigma_{y^1x} & \Sigma_{y^1y^1} & \Sigma_{y^1y^2} \\\Sigma_{y^2x} & \Sigma_{y^2y^1} & \Sigma_{y^2y^2} },  \]
\beq{\label{eqn:sigmaij}\text{where}\qquad \Sigma_{y^iy^j} = \Sigma_{y^iy^i}^{\frac{1}{2}}\Sigma_{y^jy^j}^{\frac{1}{2}\transpose},\qquad \Sigma_{y^ix} = \Sigma_{y^iy^i}^{\frac{1}{2}}\Sigma_{xx}^{\frac{1}{2}\transpose} \quad\text{for } i,j=1,2.}
Since $\VSF X, \VSF Y^1, \VSF Y^2$ are jointly Gaussian random variables, the conditional expectations $\mathds{E}[\VSF X|\VSF Y^i]$ and $\mathds{E}[\VSF Y^{-i}|\VSF Y^i] $ are  affine functions of $\VSF Y^i$, given by
\beq{\mathds{E}[\VSF X|\VSF Y^i] &=& \VSF m_x + \Sigma_{xy^i}\Sigma_{y^iy^i}^{-1}(\VSF Y^i-\VSF m_{y^i}),\\
\mathds{E}[\VSF Y^{-i}|\VSF Y^i] &=& \VSF m_{y^{-i}} + \Sigma_{y^{-i}y^i}\Sigma_{y^iy^i}^{-1}(\VSF Y^i-\VSF m_{y^i}),}
where $\Sigma_{y^iy^i}^{-1}$ is the generalized inverse (pseudo inverse) of $\Sigma_{y^iy^i}$ \cite{catlin1989} for $i=1,2$.

The cost functions are
\beq{c^1(\VSF X,\VSF U^1,\VSF U^2) &=& \bmat{\VSF X^\transpose,\VSF U^{1\transpose},\VSF U^{2\transpose}}C\bmat{\VSF X \\\VSF U^1\\\VSF U^2} + 2\bmat{d_1, d_2, d_3}\bmat{\VSF X \\\VSF U^1\\\VSF U^2}+ r^1 \label{eqn:j1aux},\\
 c^2(\VSF X,\VSF U^1,\VSF U^2) &=& \bmat{\VSF X^\transpose,\VSF U^{1\transpose},\VSF U^{2\transpose}}E\bmat{\VSF X \\\VSF U^1\\\VSF U^2} + 2\bmat{f_1, f_2, f_3}\bmat{\VSF X \\\VSF U^1\\\VSF U^2} + r^2,\label{eqn:j2aux}\\
\text{where}\quad C &=& \bmat{C_{11} &C_{12} &C_{13}\\C_{12}^{\transpose} &C_{22} &C_{23}\\C_{13}^{\transpose} &C_{23}^{\transpose} & C_{33}} \quad\text{ and }\quad E = \bmat{E_{11} &E_{12} &E_{13}\\E_{12}^{\transpose} &E_{22} &E_{23}\\E_{13}^{\transpose} &E_{23}^{\transpose} & E_{33}}.\nonumber}
with $C\geq 0, E\geq 0$, $C_{22}>0, E_{33}>0$, $C_{ij},E_{ij}$ are matrices, $d_i,f_i$ are row vectors of appropriate dimensions and $r_1,r_2$ are scalar constants. 

Controller $i$ observes $\VSF Y^i, i=1,2$ and selects $\VSF U^i$ according to a decision rule $g^i$, that is, $\VSF U^i = g^i(\VSF Y^i)$, where $g^i$ is a measurable function of $Y^i$ satisfying $\mathds{E}[g^{i\transpose}(\VSF Y^{i})g^{i}(\VSF Y^{i})] < \infty$. Let the space of all such measurable functions $g^i$ be denoted by $\ALP A\ALP G^i$. This game is referred to as game {\bf AG1}$(\VSF m,\Sigma,c^1,c^2)$. We make the following assumption on the cost functions of the controllers in {\bf AG1}$(\VSF m,\Sigma,c^1,c^2)$.

\begin{assumption}\label{assm:eigcost}
For any square matrix $A$, let $\bar{\lambda}(A)$ denotes the positive square root of the maximum eigenvalue of $A^{\transpose}A$. For the matrix tuple $(C,E)$, define $K_1 = C_{22}^{-1}C_{23}E_{33}^{-1}E_{23}^{\transpose}$ and $K_2 = E_{33}^{-1}E_{23}^{\transpose}C_{22}^{-1}C_{23}$. Let $\ALP K_i$ be the space of all matrices that are similar to $K_i,\: i=1,2$, that is, $\tilde K\in\ALP K_i$ implies there exists a square invertible matrix $L$ of appropriate dimensions such that $\tilde K = LK_iL^{-1}$. There exists an $i_0\in\{1,2\}$ and a matrix $K\in \ALP K_{i_0}$ such that $\bar{\lambda}(K)<1$.{\hfill$\Box$}
% \beq{\label{eqn:auxsuffcond} \bar{\lambda}(K)<1.}
% \beq{\label{eqn:auxsuffcond}\bar{\lambda}(C_{22}^{-1}C_{23}E_{33}^{-1}E_{23}^{\transpose}) < 1 \quad \text{or} \quad \bar{\lambda}(E_{33}^{-1}E_{23}^{\transpose}C_{22}^{-1}C_{23}) < 1.} {\hfill$\Box$}
\end{assumption}

In the next lemma, which builds on and follows from the earlier results in \cite{Basar1975,Basarmulti}, we show that the Nash equilibrium of the auxiliary game {\bf AG1}$(\VSF m,\Sigma,c^1,c^2)$ that satisfies Assumption \ref{assm:eigcost} exists in the space $\ALP A\ALP G^1\times\ALP A\ALP G^2$, is unique, and is affine in the information of the controllers.

%  It should be noted that a similar result has also been proved in \cite{Basarmulti}.
 
\begin{lemma}\label{lem:basargame}
The following statements hold:
\begin{enumerate} 
\item For the 1-stage game \textbf{AG1}, a pair of decision rules $g^{1\star}, g^{2\star}$ is a Nash equilibrium if and only if they simultaneously satisfy the following two equations,
\beqq{g^{1\star}(\VSF Y^1) &=& -C_{22}^{-1}d_2^{\transpose} - C_{22}^{-1}C_{12}^\transpose\mathds{E}[\VSF X|\VSF Y^1] - C_{22}^{-1}C_{23}\mathds{E}[g^{2\star}(\VSF Y^2)|\VSF Y^1], \nonumber \\
g^{2\star}(\VSF Y^2) &=& -E_{33}^{-1}f_3^{\transpose} - E_{33}^{-1}E_{13}^\transpose\mathds{E}[\VSF X|\VSF Y^2] - E_{33}^{-1}E_{23}^{\transpose}\mathds{E}[g^{1\star}(\VSF Y^1)|\VSF Y^2].}

\item If the matrices $(C,E)$ in the cost functions of game {\bf AG1}$(\VSF m,\Sigma,c^1,c^2)$ satisfy Assumption \ref{assm:eigcost}, then the game has a unique Nash equilibrium in the class of all Borel measurable strategies $\ALP A\ALP G^1\times\ALP A\ALP G^2$, given as
\beq{\label{eqn:gistaraux} g^{i\star}(\VSF Y^i) =  T^i(\VSF Y^i-\VSF m_{y^i}) + b^i,}
where $b^1,b^2$ are 
solutions of the following pair of equations
\beqq{b^1 &=& -C_{22}^{-1}[d_2^{\transpose}+C_{12}\VSF m_x+C_{23}b^2], \\
b^2 &=& -E_{33}^{-1}[f_3^{\transpose}+E_{13}\VSF m_x +E_{23}^{\transpose}b^1 ],}
and are of the form
\beq{\label{eqn:b1b2}b^1 &=& l^1+L^1 \VSF m_x,\qquad  b^2 = l^2+L^2\VSF m_x,}
and $T^1, T^2$ are solutions of  the following pair of equations
\beq{\label{eqn:t1}T^1 &=& -C_{22}^{-1}[C_{12}^\transpose\Sigma_{xy^1}\Sigma_{y^1y^1}^{-1}+ C_{23}T^2\Sigma_{y^2y^1}\Sigma_{y^1y^1}^{-1}], \\
\label{eqn:t2}T^2 &=& -E_{33}^{-1}[E_{13}^\transpose\Sigma_{xy^2}\Sigma_{y^2y^2}^{-1} +E_{23}^{\transpose}T^1\Sigma_{y^1y^2} \Sigma_{y^2y^2}^{-1}].}
Here, $l^i $ and $L^i$ are independent of $\VSF m$ for both $i=1,2$. 

% If the matrices $(I-C_{22}^{-1}C_{23}E_{33}^{-1}E_{23}^{\transpose})$ and $(I-E_{33}^{-1}E_{23}^{\transpose}C_{22}^{-1}C_{23})$ are invertible, then $l^i $ and $L^i, i=1,2$ are given by
% \beqq{ l^1 &=& (I-C_{22}^{-1}C_{23}E_{33}^{-1}E_{23}^{\transpose})^{-1} \left(C_{22}^{-1}[C_{23}E_{33}^{-1}f_3^\transpose-d_2^\transpose ] \right),\\
% L^1 & = & (I-C_{22}^{-1}C_{23}E_{33}^{-1}E_{23}^{\transpose})^{-1} \left(C_{22}^{-1}[C_{23}E_{33}^{-1} E_{13}^\transpose -C_{12}^\transpose ] \right),\\
% l^2 &=& (I-E_{33}^{-1}E_{23}^{\transpose}C_{22}^{-1}C_{23})^{-1} \left(E_{33}^{-1}[E_{23}^\transpose C_{22}^{-1}d_2^\transpose-f_3^\transpose ] \right),\\
% L^2 & = & (I-E_{33}^{-1}E_{23}^{\transpose}C_{22}^{-1}C_{23})^{-1} \left(E_{33}^{-1}[E_{23}^\transpose C_{22}^{-1} C_{12}^\transpose -E_{13}^\transpose ] \right).}
\item The expected costs to the controllers when they play according to Nash equilibrium $(g^{1\star},g^{2\star})$ are
\beq{\label{eqn:auxeqcost}\mathds{E}[c^i(\VSF X,g^{1\star}(\VSF Y^1),g^{2\star}(\VSF Y^2))] = \VSF m^{\transpose} \Phi^i {\VSF m} + \Xi^{i}\VSF m + \Upsilon^i,}
where the matrices $\Phi^i$, $\Xi^i$ and $\Upsilon^i$ for $i=1,2$ are defined by
\beqq{&\tilde{L} := \bmat{0 &0 & 0\\ L^1 & -T^1 & 0\\ L^2 & 0 & -T^2}, \quad \tilde{T} := \bmat{I &0 & 0\\ 0 & T^1 & 0\\0& 0 & T^2}, \quad
\tilde{l} := \bmat{0 \\ l^1 \\l^2},\\
&\Phi^1 = (\tilde{T}+\tilde{L})^{\transpose}C(\tilde{T}+\tilde{L}), \qquad \Xi^1 = 2\tilde{l}^{\transpose}C(\tilde{T}+\tilde{L})+2[d_1,d_2, d_3](\tilde{T}+\tilde{L}),\\
&\Phi^2 = (\tilde{T}+\tilde{L})^{\transpose}E(\tilde{T}+\tilde{L}), \qquad \Xi^2 = 2\tilde{l}^{\transpose}E(\tilde{T}+\tilde{L})+2[f_1,f_2, f_3](\tilde{T}+\tilde{L}),\\
&\Upsilon^1 = r^1 + 2[d_1,d_2, d_3]\tilde{l} + \textrm{trace}\left(\tilde{T}^{\transpose}C\tilde{T}\Sigma\right)+2\tilde{l}^{\transpose}C\tilde{l}, \\
& \Upsilon^1 = r^2 + 2[f_1,f_2, f_3]\tilde{l} + \textrm{trace}\left(\tilde{T}^{\transpose}E\tilde{T}\Sigma\right)+2\tilde{l}^{\transpose}E\tilde{l}.}
\end{enumerate}
\end{lemma}
\begin{proof}
Part 1 of the lemma is proved by differentiating $\ex{c^i(\VSF X,\VSF U^i, g^{-i}(\VSF Y^{-i}))|\VSF Y^i}$ with respect to $\VSF U^i$ and setting it equal to zero. For the proof of Part 2 of the lemma, see Appendix \ref{sec:auxiliary}. For proving Part 3 of the lemma, notice that if $\VSF U^i = g^{i\star}(\VSF Y^i)$ for $i=1,2$, then
\beqq{\bmat{\VSF X \\\VSF U^1\\\VSF U^2} = \tilde{T}\bmat{\VSF X \\\VSF Y^1\\\VSF Y^2}+\tilde{L}\VSF m+\tilde{l}.}
Now, substituting this in the expressions for $c^1$ and $c^2$ and taking the expectations, we get the expected costs of the controllers. This completes the proof of the lemma.
\end{proof}

\begin{remark}\label{rem:taui}
The Nash equilibrium strategy of Player $i$ given in \eqref{eqn:gistaraux} of the auxiliary game {\bf AG1} can be rewritten as
\beqq{g^{i\star}(\VSF Y^i) = \matrix{ccc}{l^i+L^i\VSF m_x-T^i\VSF m_{y^i} &|& T^i } \matrix{c}{1 \\ \VSF Y^i}.}
It should also be noted that the unique Nash equilibrium in the auxiliary game {\bf AG1} exists in the class of all Borel measurable strategies of the controllers. {\hfill$\Box$}
\end{remark}

In the next subsection, we consider a class of dynamic LQG games satisfying certain assumptions, and show that each LQG game in that class admits a unique common information based Markov perfect equilibrium.

\subsection{Generalization to Dynamic LQG Games}\label{sub:generallqg}
In this subsection, we consider LQG games that satisfy Assumptions \ref{assm:lqginfoevolution} and \ref{assm:lqgseparation}. In order to prove the main result of the section, we need the following lemma.

\begin{lemma}\label{lem:bayesiangamet}
Consider an LQG game {\bf G1} that satisfies Assumptions \ref{assm:lqginfoevolution} and \ref{assm:lqgseparation}. Fix a time step $t\in\{1,\ldots,T-2\}$. If the expected value functions $V^i_t, i=1,2$ of the controllers at time $t+1$ are affine-quadratic functions of $\VEC m_{t+1}$, then the one-stage Bayesian game at time step $t$ with any mean $\VEC m_t\in\ALP S_t$ is an instance of auxiliary game {\bf AG1}.
\end{lemma}
\begin{proof}
Consider a time step $t\in\{1,\ldots,T-1\}$ and a realization $\VEC c_t$ of the common information at time step $t$. Define
\beqq{\VSF X = \VEC S_t = \bmat{\VEC X_t \\ \VEC P^1_t\\ \VEC P^2_t},\quad \VSF Y^i = \VEC P^i_t,\quad \VSF U^i = \VEC U^i_t, \qquad i=1,2.}
The probability measure on the state $\VSF X$ is taken to be equal to the common information based conditional measure $\pi_t(d\VEC s_t) = \pr{d\VEC s_t|\VEC c_t}$, that admits a Gaussian density function with mean $\VEC m_t$ (dependent on $\VEC c_t$) and variance $\Sigma_t$, which are defined in Lemma \ref{lemma:lqgmean}. The observation $\VSF Y^i$ of auxiliary Controller $i$ is the private information $\VEC P^i_t$.

We now prove that the one-stage Bayesian game defined above is an instance of the Auxiliary game {\bf AG1}. First, note that $\VSF Y^i = H^i\VSF X$ for some appropriate matrix $H^i$, $i\in\{1,2\}$, which implies that the assumption on the covariance matrix of the auxiliary game given in \eqref{eqn:sigmaij} is satisfied by the auxiliary game defined above. We just need to verify that the cost functions of the controllers of the one-stage Bayesian game are of the same form as \eqref{eqn:j1aux} and \eqref{eqn:j2aux}. 

At any time step $t\leq T-1$, let us assume that the expected value function of the Controller $i$ is $V^i_{t+1}(\VEC m_{t+1}) = \VEC m^{\transpose}_{t+1}\Phi^i_{t+1}\VEC m_{t+1}+\Xi^i_{t+1}\VEC m_{t+1}+\Upsilon^i_{t+1}$ for some appropriate positive definite matrix $\Phi^i_{t+1}$, matrix $\Xi^i_{t+1}$ and a non-negative real number $\Upsilon^i_{t+1}$. Recall from Lemma \ref{lemma:lqgmean} that $\VEC M_{t+1}:= F^1_t(\VEC m_t,\VEC Z_{t+1})$, where $F^1_t$ is an affine function of $\VEC m_t$ and $\VEC Z_{t+1}$. Thus, the cost-to-go for the Controller $i\in\{1,2\}$ in the one-stage Bayesian game at time step $t< T$ with mean $\VEC m_t$ is of the form
\beq{\label{eqn:citlqg}\check{\check{c}}^i_{t}(\VEC m_t;\VEC S_t,\VEC U^1_t,\VEC U^2_t,\VEC Z_{t+1}) &:=& c^i_t(\VEC X_t,\VEC U^1_t,\VEC U^2_t)+(F^{1}_t(\VEC m_t,\VEC Z_{t+1}))^\transpose \Phi^i_{t+1} F^1_t(\VEC m_t,\VEC Z_{t+1})\nonumber\\ & & +\Xi^i_{t+1} F^{1}_t(\VEC m_t,\VEC Z_{t+1})+ \Upsilon^i_{t+1}.\qquad}
Now, recall the definition of $\VEC Z_{t+1}$ in Assumption \ref{assm:lqginfoevolution}, and substitute for $\VEC Y^1_{t+1}$ and $\VEC Y^2_{t+1}$ in the expression for $\VEC Z_{t+1}$ in terms of $\VEC X_t$, $\VEC U^1_t$, $\VEC U^2_t$ and noises using \eqref{eq:lqgdynamics} and \eqref{eq:lqgobserv}. Thus, $\VEC Z_{t+1}$ is an affine map of $\VEC S_t$, $\VEC U^1_t$, $\VEC U^2_t$ and noises $\VEC W^0_t,\VEC W^1_{t+1}$ and $\VEC W^2_{t+1}$. Also recall from Lemma \ref{lemma:lqgmean} that $F^1_t$ is an affine map of its arguments. Define $\bar{c}^i_t, i=1,2$ as
\beqq{\bar{c}^i_{t}(\VEC m_t;\VEC S_t,\VEC U^1_t,\VEC U^2_t) = \ex{\check{\check{c}}^i_{t}(\VEC m_t;\VEC S_t,\VEC U^1_t,\VEC U^2_t,\VEC Z_{t+1})|\VEC S_t,\VEC U^1_t,\VEC U^2_t}.}
Thus, the expression for cost function $\bar{c}^i_t$, $i=1,2$ are precisely of the forms
\beq{\label{eqn:barc1t}\bar{c}^1_{t}(\VEC m_t;\VEC S_t,\VEC U^1_t,\VEC U^2_t) &=& \bmat{\VEC S_t^\transpose,\VEC U_t^{1\transpose},\VEC U_t^{2\transpose}}C_t\bmat{\VEC S_t \\\VEC U_t^1\\\VEC U_t^2} + 2\VEC m_t^\transpose D_t \bmat{\VEC S_t \\\VEC U_t^1\\\VEC U_t^2}+ r^1_t \VEC m_t+\tilde\Upsilon^1_{t+1},\qquad\\
\label{eqn:barc2t}\bar{c}^2_{t}(\VEC m_t;\VEC S_t,\VEC U^1_t,\VEC U^2_t) &=& \bmat{\VEC S_t^\transpose,\VEC U_t^{1\transpose},\VEC U_t^{2\transpose}}E_t\bmat{\VEC S_t \\\VEC U_t^1\\\VEC U_t^2} + 2\VEC m_t^\transpose F_t \bmat{\VEC S_t \\\VEC U_t^1\\\VEC U_t^2} + r^2_t \VEC m_t +\tilde\Upsilon^2_{t+1},\qquad}
where $C_t,D_t,E_t,F_t, r^1_t,r^2_t,\tilde{\Upsilon}^1_{t+1}, \tilde{\Upsilon}^2_{t+1}$ are dependent on matrices $R^i$, $\Phi^i_{t+1}$, $\Xi^i_{t+1}$ for $i=1,2$, the linear map $F^1_t$, the variances of noises $\VEC W^0_t,\VEC W^1_{t+1}$ and $\VEC W^2_{t+1}$, and the projection function $\zeta_{t+1}$, where $\zeta_{t+1}$ is defined in \eqref{eq:lqgcommoninfo}. The cost functions of the controllers given above are of the same form as considered in \eqref{eqn:j1aux} and \eqref{eqn:j2aux}. Thus, the one-stage Bayesian game at time $t$ with mean $\VEC m_t$ is the same as auxiliary game {\bf AG1}$(\VEC m_t,\Sigma_t,\bar{c}^1_t(\VEC m_t;\cdot),\bar{c}^2_t(\VEC m_t;\cdot))$. This completes the proof of the lemma.
\end{proof}

\begin{definition}\label{def:agtg1}
For any LQG game {\bf G1}, the corresponding one-stage Bayesian game at time step $t$ with mean $\VEC m_t$ and cost functions of the controllers given by \eqref{eqn:barc1t} and \eqref{eqn:barc2t} is referred to as ${\bf AG}_t({\bf G1};\VEC m_t)$. {\hfill$\Box$}
\end{definition}

Lemma \ref{lem:bayesiangamet} implies that the one-stage Bayesian game of the LQG game {\bf G1} at time $t$ with mean $\VEC m_t$ is an instance of an auxiliary game {\bf AG1}$(\VEC m_t,\Sigma_t,\bar{c}^1_{t}(\VEC m_t;\cdot),\bar{c}^2_{t}(\VEC m_t;\cdot))$, where $\bar{c}^i_{t}$ is defined in \eqref{eqn:citlqg}. If the matrix tuple $(C_t,E_t)$, as defined in \eqref{eqn:barc1t} and \eqref{eqn:barc2t}, satisfies Assumption \ref{assm:eigcost}, then for any realization $\VEC m_t$, Lemma \ref{lem:basargame} implies that there exists a unique Nash equilibrium of the one-stage Bayesian game at time step $t$, which is affine in the private information of the controllers. Furthermore, the expected equilibrium costs are affine-quadratic in the mean $\VEC m_t$. This crucial observation about LQG games leads us to the next theorem, which is also the main result of this section. First, we need the following assumption on the cost functions of the one-stage Bayesian games of game {\bf G1}. 

\begin{assumption}\label{assm:lqgcost} 
At every time step $t\in\{1,\ldots,T-1\}$ of game {\bf G1}, the matrix tuple $(C_t,E_t)$, as defined in \eqref{eqn:barc1t} and \eqref{eqn:barc2t} and obtained using the procedure outlined in Algorithm 1 in Subsection \ref{sub:stagebayesian}, satisfies Assumption \ref{assm:eigcost}. {\hfill$\Box$}
\end{assumption}

The main result of this section is now captured by the following theorem.

\begin{theorem}\label{thm:lqgg1}
Consider an LQG game {\bf G1} that satisfies Assumptions \ref{assm:lqginfoevolution} and \ref{assm:lqgseparation}. If Assumption \ref{assm:lqgcost} also holds for game {\bf G1}, then it admits a unique common information based Markov perfect equilibrium in the class of {\it all Borel measurable strategies} of the controllers. Furthermore, the equilibrium strategy of Controller $i\in\{1,2\}$ is affine in its information at all time steps.
\end{theorem}
\begin{proof}
We follow the procedure outlined in Algorithm 1 in Subsection \ref{sub:stagebayesian}. At time step $T-1$, let $\VEC m_{T-1}$ be a realization of the common information based conditional mean. Consider the one-stage Bayesian game ${\bf AG}_{T-1}({\bf G1};\VEC m_{T-1})$ at time $T-1$. Since Assumption \ref{assm:lqgcost} holds, we use the result of Lemma \ref{lem:basargame} to conclude that a unique Nash equilibrium policies of the controllers exist. Furthermore, the Nash equilibrium policies of the one-stage Bayesian game ${\bf AG}_{T-1}({\bf G1};\VEC m_{T-1})$ are affine in the conditional mean and the private information of the controllers at that time step. Since the conditional mean $\VEC m_{T-1}$ is affine in the common information of the controllers, we conclude that the Nash equilibrium policies of the controllers in the one-stage Bayesian game ${\bf AG}_{T-1}({\bf G1};\VEC m_{T-1})$ is affine in the information of the controllers. 

We continue this process for all possible means $\VEC m_t\in\ALP S_t$ at all time steps $t\in\{T-1,T-2,\ldots,1\}$ to conclude that there is a unique common information based Markov perfect equilibrium for game {\bf G1}.
\end{proof}

\begin{remark}
It should be noted that one can compute the Bayesian Nash equilibrium of game {\bf AG1} simply by solving a set of linear equations. In the dynamic LQG game {\bf G1} satisfying the sufficient conditions of Theorem \ref{thm:lqgg1}, the controllers can just solve a set of linear equations at successive time steps to obtain the unique common information based Markov perfect equilibrium; thus, computing the equilibrium is inexpensive in the class of LQG games.{\hfill $\Box$}
\end{remark}

\subsection{LQG Games not satisfying Assumption \ref{assm:lqgcost}}\label{sub:not}
In this subsection, we show that even if an LQG game does not satisfy Assumption \ref{assm:lqgcost}, it may still admit a common information based Markov perfect equilibrium under some mild conditions. However, we cannot claim uniqueness of that equilibrium like we did in Theorem \ref{thm:lqgg1} in the previous subsection.

Consider the auxiliary game {\bf AG1} discussed above. We needed Assumption \ref{assm:eigcost} in two places in the result of Lemma \ref{lem:basargame} above - (i) to conclude the uniqueness of Nash equilibrium provided that it exists, and (ii) to show the existence of matrices $l^1,l^2, L^1, L^2,T^1$ and $T^2$ as defined in \eqref{eqn:b1b2}, \eqref{eqn:t1} and \eqref{eqn:t2} above. However, if we drop this assumption and instead make a milder assumption, then we can obtain a result that is weaker than what we got above. First, we state the assumption we make to obtain the weaker result.  

\begin{assumption}\label{assm:eigcostb}
In game {\bf AG1}, the either $(I - C_{22}^{-1}C_{23}E_{33}^{-1}E_{23}^{\transpose})$ is invertible or $(I-E_{33}^{-1}E_{23}^{\transpose}C_{22}^{-1}C_{23})$ is invertible, and there exists a unique solution to the coupled pair of equations \eqref{eqn:t1} and \eqref{eqn:t2}. {\hfill$\Box$}
\end{assumption}

% Note that the if an auxiliary game {\bf AG1} satisfies Assumption \ref{assm:eigcost}, then it also satisfies Assumption \ref{assm:eigcostb}.
If we make this assumption on the auxiliary game {\bf AG1}, then we can conclude the following about the Nash equilibrium of the game.
  
\begin{lemma}\label{lem:basargameb}
If the auxiliary game {\bf AG1}$(\VSF m,\Sigma,c^1,c^2)$ satisfies Assumption \ref{assm:eigcostb}, then the game admits a Nash equilibrium. The expressions of Nash equilibrium control laws and expected costs are the same as in Lemma \ref{lem:basargame}. 
\end{lemma}
\begin{proof}
The proof is analogous to the proof of Lemma \ref{lem:basargame}.
\end{proof}

This brings us to the following result for LQG games that may not satisfy Assumption \ref{assm:lqgcost}.
\begin{theorem}\label{thm:lqgg2}
Consider an LQG game {\bf G1} that satisfies Assumptions \ref{assm:lqginfoevolution} and \ref{assm:lqgseparation}. For all time steps $t\in\{1,\ldots,T-1\}$ and realizations of the mean $\VEC m_t\in\ALP S_t$, obtain the one-stage Bayesian game ${\bf AG}_t({\bf G1};\VEC m_t)$ by following the steps of Algorithm 1 in Subsection \ref{sub:stagebayesian}. If ${\bf AG}_t({\bf G1};\VEC m_t)$ satisfies Assumption \ref{assm:eigcostb} for all $t\in\{1,\ldots,T-1\}$ and $\VEC m_t\in\ALP S_t$, then game {\bf G1} admits a common information based Markov perfect equilibrium.
\end{theorem}
\begin{proof}
The proof follows from the same arguments as in the proof of Theorem \ref{thm:lqgg1}. It should be noticed that for a fixed affine strategy of Controller 1, the one-person dynamic optimization problem for Controller 2 can be solved using a dynamic programming (existence follows from Assumption \ref{assm:eigcostb}) to obtain optimal strategies that are linear in its information. 
\end{proof}

Notice that we do not claim uniqueness of the common information based Markov perfect equilibrium for LQG games that do not satisfy Assumption \ref{assm:lqgcost}.
\subsection{An Illustrative Example}\label{sub:global}
In this section, we consider an example of a two-player non-zero sum game considered above. There are three states in the game, out of which one is a global state that is observed by both controllers, and two states are local states of the controllers. The state evolution is given by
\beqq{\VEC X^0_{t+1} &=& A \VEC X^0_t+B^1\VEC U^1_t+B^2\VEC U^2_t+\VEC W^{0}_t,\\
\VEC X^i_{t+1} &=& A^i \VEC X^0_t+D^{i}_1\VEC U^1_t+D^i_2 \VEC U^2_t+\VEC W^i_t,\quad i=1,2,}
where $\VEC X^0_t$ is the global state and $\VEC X^i_t,i=1,2$ are the local states for $t=1,\ldots,T-1$. The noise processes $\VEC W^i_t$ are assumed to be mutually independent mean-zero Gaussian random variables with variances $\Lambda^i_t$ for $i=0,1,2$ and $t=1,\ldots,T-1$. Player $i$'s information $\VEC I^i_t$ lies in a Euclidean space, that is $\ALP I^i_t = \ALP X^0_{1:t}\times\ALP X^i_{1:t}\times\ALP U^{1:2}_{1:t-1}$. The cost functions of the players are given by
\beqq{c^i_T(\VEC x^0_T,\VEC x^i_T) &=& \VEC x^{0\transpose}_T Q^i \VEC x^0_T+\VEC x^{i\transpose}_T Q^i \VEC x^i_T,\\
c^i_t(\VEC x^0_t,\VEC x^i_t,\VEC u^{1:2}_t) &=& \VEC x^{0\transpose}_t Q^i \VEC x^0_t+\VEC x^{i\transpose}_t Q^i \VEC x^i_t+\VEC u^{1\transpose}_{t}R^i\VEC u^1_t+\VEC u^{2\transpose}_{t}S^i\VEC u^2_t, \quad i=1,2.}
We first verify in the following lemma that both assumptions on the information structure is satisfied by this game. Toward this end, recall that the space of common information of the players is $\ALP C_t:=\ALP X^0_{1:t}\times\ALP U^{1:2}_{1:t-1}$ and the space of private information of Player $i$ is $\ALP P^i_t:=\ALP X^i_{1:t}$.
\begin{lemma}
\label{lem:global}
The information structure of the controllers in the game defined above satisfies Assumptions \ref{assm:lqginfoevolution} and \ref{assm:lqgseparation}.
\end{lemma}
\begin{proof}
Since the information structure is nested, Assumption \ref{assm:lqginfoevolution} is automatically satisfied. Given the common information $\VEC c_t$, it is easy to verify that the joint distribution of private informations of the players are Gaussian and independent of past strategies of the players since $\VEC W^i_t$ are mutually independent Gaussian random variables for all time steps, and $i=1,2$. This implies that Assumption \ref{assm:lqgseparation} is also satisfied.
\end{proof}

We are interested in computing common information based Markov perfect equilibrium of this game. At a time step $t$, since the private state of both controllers are affected by the global state at time step $t-1$ but not the private states at the previous time steps, the past realizations of private states of Controller $i$, $\VEC x^i_1,\ldots,\VEC x^i_{t-1}$, do not affect the common information based Markov perfect equilibrium. Therefore, for ease of exposition, we make minor changes in the notations from what we have used in previous sections. Let us denote the mean of the random variable $\VEC X^i_{t}$ given the common information $\VEC C_t$ as $\VEC M^i_{t}, i=1,2$. Note that the mean is given by $\VEC M^i_{t} =  A^i \VEC X^0_{t-1}+D^{i}_1\VEC U^1_{t-1}+D^i_2 \VEC U^2_{t-1}$. The variance of $\VEC X^i_t$ given the common information $\VEC C_t$ is $\Lambda^i_t$. Define $\VEC M_t := [\VEC X^{0\transpose}_t,\VEC M^{1\transpose}_t,\VEC M^{2\transpose}_t]^\transpose$ to be the conditional mean of the states 
$[\VEC X^{0\transpose}_t,\VEC X^{1\transpose}_t,\VEC X^{2\transpose}_t]^\transpose$ given the common information $\VEC C_t$. Also note that $\VEC Z_{t+1}=[\VEC U^{1\transpose}_t, \VEC U^{2\transpose}_t,\VEC X^{0\transpose}_{t+1}]^\transpose$. The evolution of $\VEC M_t$ is given by
\beq{\VEC M_{t+1} &=& \matrix{c}{\VEC X^0_{t+1}\\A^1 \VEC X^0_t+D^{1}_1\VEC U^1_t+D^1_2 \VEC U^2_t\\ A^2 \VEC X^0_t+D^{2}_1\VEC U^1_t+D^2_2 \VEC U^2_t}\nonumber\\
\label{eqn:f1tglobal}& =: & F^1_t(\VEC X^0_t,\VEC U^1_t,\VEC U^2_t,\VEC X^0_{t+1}).}
The conditional covariance matrix is $\Sigma_{t+1} := \text{diag}\{0,\Lambda^1_{t+1},\Lambda^2_{t+1}\}$. Let us assume that this game satisfies Assumption \ref{assm:lqgcost}. By Theorem \ref{thm:lqgg1}, we conclude that this game admits a unique common information based Markov perfect equilibrium. Now, we compute the unique common information based Markov perfect equilibrium of this game.

\begin{enumerate}
\item At the terminal time $T$, for each realization $\VEC m:=[\VEC x^{0\transpose}_T, \VEC m^{1\transpose},\VEC m^{2\transpose}]^\transpose$ of $\VEC M_{T}$, the one-stage Bayesian game  $SG_{T}(\VEC m)$ is defined as follows 
\begin{enumerate}
\item The conditional probability distribution on $\ALP X^0_{T}\times \ALP X^1_{T}\times\ALP X^2_{T}$ given the common information $\VEC c_T$ is a Gaussian density with mean $\VEC m$ and covariance $\Sigma_{T}$.
\item Player $i$ observes $\VEC X^0_{T}, \VEC X^i_{T}$, $i=1,2$. No action is chosen.
\item Player $i$'s cost is $c^i_T(\VEC x^0_T,\VEC x^i_T)$. 
\end{enumerate}
The expected costs as functions of beliefs are of the form $V^i_{T}(\VEC m) =\VEC x^{0\transpose}_T Q^i \VEC x^0_T +\VEC m^{i\transpose} Q^i\VEC m^i +\textrm{trace}(Q^i\Lambda^i_T)$.

\item At time $t<T$, for each realization $\VEC m$ of $\VEC M_{t}$, we define a one-stage Bayesian game $SG_{t}(\VEC m)$ where 
\begin{enumerate}
\item The probability distribution on $\ALP X^0_{t}\times \ALP X^1_{t}\times\ALP X^2_{t}$ is a Gaussian density with mean $\VEC m$ and covariance $\Sigma_{t}$.
\item  Player $i$ observes $\VEC X^0_{t}, \VEC X^i_{t}$ and chooses action $\VEC U^i_{t}$, $i=1,2$. 
\item  The cost of the subgame $SG_t(\VEC m)$ accrued by Player $i$ is 
\beqq{c^i_t(\VEC x^0_t,\VEC x^i_t,\VEC u^{1:2}_t) + \ex{V^i_{t+1}(F^1_{t}(\VEC x^0_t,\VEC u^1_{t},\VEC u^2_{t},\VEC X^0_{t+1}))},}
where $F^1_{t}$ is defined in \eqref{eqn:f1tglobal}. If we expand this cost function and write it in terms of $\VEC x^0_t,\VEC u^1_t$, $\VEC u^2_t$ and noise variables, then we observe that the resulting cost function is of the same form as in the auxiliary game {\bf AG1} considered in Subsection \ref{sec:basargame} above. The Nash equilibrium of the Bayesian game $SG_t(\VEC m)$ is computed using the result in Lemma \ref{lem:basargame}. The value functions of the virtual players are of the form $V^i_{t}(\VEC m) =\VEC x^{0\transpose}_t \Phi^i_t \VEC x^0_t +\VEC m^{i\transpose} \Xi^i_t \VEC m^i + \Upsilon^i_t$, where $\Phi^i_t$ and $\Xi^i_t$ are non-negative definite matrices and $\Upsilon^i_t$ is a non-negative real number.
\end{enumerate}
 Thus, we computed the unique common information based Markov perfect equilibrium of the game.
\end{enumerate}
In the next section, we show that there could be several other Nash equilibria of game {\bf G1}.

\section{A Game with Multiple Nash Equilibria}
\label{sec:information}
In previous sections, we outlined an algorithm that can be used to compute Nash equilibria of games that satisfy Assumptions \ref{assm:lqginfoevolution} and \ref{assm:lqgseparation}. In this section, we exhibit an example of a game of asymmetric information that has several Nash equilibria and a unique common information based Markov perfect equilibrium. This reinforces Remark \ref{rem:multiequilibria} stated after Theorem \ref{thm:mpevirtual}, which points out that our algorithm computes only a subclass of those Nash equilibria that can be obtained using the Markov perfect equilibrium of the corresponding symmetric information game between the virtual players.

To illustrate the existence of multiple Nash equilibria, we follow the lines in Example 1 of \cite[p. 241]{Basarmulti}, and consider the following two-stage game in which all variables are scalar:
\beqq{x_2 &=& x_1+u^1_1+u^2_1+w^0_1, \quad y^1_1 = x_1+w^1_1,\quad y^2_1 = x_1+w^2_1,\\
x_3 &=& x_2+u^2_2+w^0_2,\qquad\quad y^2_2 = x_2+w^2_2.} 
The primitive random variables $\{X_1,W^0_1,W^0_2,W^1_1,W^2_1,W^2_2\}$ are all mutually independent, mean zero Gaussian random variables with unit variance. The information structure of Controller $1$ is $\VEC I^1_1 = Y^1_1$ and $\VEC I^1_2 = [U^1_1,U^2_1,Y^1_1,Y^2_1]^\transpose$, and the information structure of Controller $2$ is $\VEC I^2_1 = Y^2_1$ and $\VEC I^2_2 = [Y^2_2,U^1_1,U^2_1,Y^1_1,Y^2_1]^\transpose$. Since the information structure is of the one-step delayed sharing pattern type, this game satisfies Assumptions \ref{assm:lqginfoevolution} and \ref{assm:lqgseparation}. The cost functions of the controllers are
\beqq{J^1(g^1_1,g^2_1,g^2_2) &=& \ex{(X_3)^2+(U^1_1)^2},\\
J^2(g^1_1,g^2_1,g^2_2)  &=& \ex{(X_3)^2+(U^2_1)^2+(U^2_2)^2}.}
The Kalman filter equations for evolution of common information based belief are given by
\beqq{m^0_2 &=& \ex{X_2|\VEC c_2} = \frac{y^1_1+y^2_1}{3}+u^1_1+u^2_1,\quad \Sigma^0_2 = \ex{(X_2-m^0_2)^2} = \frac{4}{3}.}
Now, we compute the common information based Markov perfect equilibrium of this game.
\subsection{Common Information based Markov Perfect Equilibrium}
Recall that since this is an LQG game, there exists a unique common information based Markov perfect equilibrium. We now apply our algorithm to compute the common information based Markov perfect equilibrium of this game. In the sub-game starting at time step $t=2$, Controller 2 is the only player acting and the cost-to-go function for the sub-game starting at $t=2$ is strictly convex in the control action of Controller 2. The unique optimal control law for Controller 2 for the sub-game at $t=2$, given its information, is given by   
\beq{\label{eqn:g2}g^{2\star}_2(\VEC i^2_2) = -\frac{1}{2}\ex{X_2|\VEC i^2_2} = -\frac{1}{2}\ex{X_2|m^0_2,y^2_2}  = -\frac{1}{2}\left(m^0_2+\frac{4}{7}(y^2_2-m^0_2)\right).}
The expected value functions of the controllers at time step $2$ are
\beqq{V^1_2(m^0_2) = \frac{1}{4}(m^0_2)^2+\frac{37}{21},\qquad V^2_2(m^0_2) = \frac{1}{2}(m^0_2)^2+\frac{41}{21}.}
Now, using the result of Lemma \ref{lem:basargame}, we obtain the unique Nash equilibrium of the sub-game starting at the first time step to be
\beq{\label{eqn:g1g2}g^{1\star}_1(y^1_1) = -\frac{5}{59} y^1_1,\qquad g^{2\star}_1(y^2_1) = -\frac{9}{59} y^2_1.}
Thus, this game has a unique common information based Markov perfect equilibrium $(g^{1\star}_1,(g^{2\star}_1,g^{2\star}_2))$. In the next subsection, we show that there exists a continuum of Nash equilibria in this game, and those equilibria cannot be obtained using our approach.

\subsection{Other Nash Equilibria}
In this subsection, we show that if Controller 2 uses the common information (instead of only mean $m^0_2$) to construct its control law at time step $t=2$, then we have a continuum of Nash equilibrium in this game.

We now define a tuple of strategies $(g^{1\dagger}_1,(g^{2\dagger}_1,g^{2\dagger}_2))\in \ALP G^1_1\times\ALP G^2_{1:2}$ of both controllers, parametrized by a real number $\lambda\neq -59/22$:
\beq{g^{1\dagger}_1(y^1_1) &=& -\frac{10\lambda +5}{22\lambda +59} y^1_1,\quad\qquad g^{2\dagger}_1(y^2_1) = -\frac{2\lambda+9}{22\lambda+59} y^2_1,\nonumber\\
\label{eqn:g2dagger} g^{2\dagger}_2(\VEC i^2_2) &=& -\frac{1}{2}\left(m^0_2+\frac{4}{7}(y^2_2-m^0_2)\right) + \lambda\left(u^1_1 -g^{1\dagger}_1(y^1_1)\right).}

% Note that this strategy is {\it not} a Markov perfect equilibrium of the game among virtual players, since it depends on the variables $U^1_1$ and $Y^1_1$ {\it in addition to} the mean $M^0_2$. Moreover, $g^{2\dagger}_2$ is also not the optimal response of Controller 2 for the sub-game starting at time step $t=2$ unless $u^1_1 =g^{1\dagger}_1(y^1_1)$ or $\lambda=0$.
% 
% The expected value functions of the controllers at time step $2$ are
% \beqq{V^1_2(\VEC i^1_2) &=& \frac{1}{4}(m^0_2)^2+\frac{37}{21}+\lambda m^0_2(u^1_1-g^{1\dagger}_1(y^1_1))+\lambda^2\left(u^1_1-g^{1\dagger}_1(y^1_1)\right)^2,\\
%  V^2_2(\VEC i^2_2) &=& \frac{1}{2}(m^0_2)^2+\frac{41}{21}+2\lambda^2\left(u^1_1-g^{1\dagger}_1(y^1_1)\right)^2.}
% % At time $t=1$, the cost-to-go function for the Controller 1 is $(u^1_1)^2+V^1_2(\VEC i^1_2)$ and for Controller 2 is $(u^2_1)^2+V^2_2(\VEC i^2_2)$. For any value of $\lambda$, we get $\lambda^2+\lambda+\frac{5}{4} = (\lambda+0.5)^2+1>0$. Thus, the coefficient of $(u^i_1)^2$ in the expression for expected cost to controller $i, i=1,2$ at time step $1$ is strictly positive. If $\lambda>-5/4$, then Assumption \ref{assm:eigcost} holds. Using the result of Lemma \ref{lem:basargame}, we get
% \beqq{g^{1\dagger}_1(y^1_1) = -\frac{10\lambda +5}{22\lambda +59} y^1_1,\qquad g^{2\dagger}_1(y^2_1) = -\frac{2\lambda+9}{22\lambda+59} y^2_1,}
% for any $\lambda>-5/4$. 

We now have the following result:
\begin{lemma}
The strategy profile $(g^{1\dagger}_1,(g^{2\dagger}_1,g^{2\dagger}_2))$ is a Nash equilibrium of the game formulated above for any value of $\lambda\neq -59/22$. 
% Moreover, $(g^{1\dagger}_1,(g^{2\dagger}_1,g^{2\dagger}_2))$ is not a sub-game perfect Nash equilibrium for any $\lambda\in(-5/4,0)\cup(0,\infty)$.
\end{lemma}
\begin{proof}
We prove that for a fixed strategy $g^{1\dagger}_1$, $(g^{2\dagger}_1,g^{2\dagger}_2)$ is the best response strategy of Controller 2 and vice versa. 

First, fix $u^1_1 = g^{1\dagger}_1(y^1_1)$. Then, $g^{2\dagger}_2$ minimizes the cost $\ex{(X_3)^2\big|u^1_1 = g^{1\dagger}_1(y^1_1), \VEC i^2_2}$ (note that $\VEC i^2_2$ contains $y^1_1$). One can then verify that $g^{2\dagger}_1$ minimizes the cost functional $J^2(g^{1\dagger}_1,(g^2_1,g^{2\dagger}_2))$. Therefore, we conclude that $J^2(g^{1\dagger}_1,(g^{2\dagger}_1,g^{2\dagger}_2))\leq J^2(g^{1\dagger}_1,(g^{2}_1,g^{2}_2))$ for any $\VEC g^2:=(g^{2}_1,g^{2}_2)\in\ALP G^2_{1:2}$.

Next, fix $u^2_1 =g^{2\dagger}_1(y^2_1), u^2_2 = g^{2\dagger}_2(\VEC i^2_2)$. Then, the cost-to-go function for Controller 1 is at time step $1$ for the fixed strategy of Controller 2 is
\beqq{\frac{1}{4}(m^0_2)^2+\frac{37}{21}+\lambda m^0_2(u^1_1-g^{1\dagger}_1(y^1_1))+\lambda^2\left(u^1_1-g^{1\dagger}_1(y^1_1)\right)^2.}
The control law $g^{1\dagger}_1$ minimizes the cost functional of Controller 1 at time step $t=1$. Thus, we conclude $J^1(g^{1\dagger}_1,(g^{2\dagger}_1,g^{2\dagger}_2))\leq J^1(g^{1}_1,(g^{2\dagger}_1,g^{2\dagger}_2))$ for any $g^1_1\in\ALP G^1_1$. 
% The second statement follows from the fact that $g^{2\dagger}_2$ is not the optimal response of Controller 2 for the sub-game starting at time step $t=2$, since \eqref{eqn:g2} is the unique control law of Controller 2 that minimizes the Controller 2's cost functional of the sub-game starting at time step 2. This completes the proof of the lemma.
\end{proof}

Thus, we have proved that, in fact, there are several Nash equilibria of this game, and the Nash equilibrium obtained using our algorithm is just one among them (notice that the common information based Markov perfect equilibrium corresponds to the choice of $\lambda =0$).

\begin{remark}
It has been shown in \cite{Basaronestep} that if the control actions are not shared among the controllers, then there exists a unique Nash equilibrium in the game formulated above. The unique Nash equilibrium is the same set of control laws given in \eqref{eqn:g2} and \eqref{eqn:g1g2} with $u^i_1$ substituted with $g^{i\star}_1(y^i_1)$ in the expression of $m^0_2$. Thus, the existence of multiple Nash equilibria in this game is due to the information available to Controller 2 at time step $2$ about the action taken by the Controller 1 at time step 1. This example illustrates that more information to controllers may be harmful in a game as it gives rise to several other Nash equilibria! {\hfill$\Box$}
\end{remark}

\begin{remark}
It should also be noted that in case control actions are not shared, Assumption \ref{assm:lqgseparation} does not hold. However, \cite{Basaronestep} proved that an LQG game with one-step delayed observation sharing pattern admits a unique Nash equilibrium. Thus, Nash equilibrium in a dynamic game of asymmetric information may exist even in the absence of Assumption \ref{assm:lqgseparation} on that game. {\hfill$\Box$}
\end{remark}

\subsection{Effects on Expected Costs}
We now compare the expected costs to the controllers, if the Controller 2 plays according to a Nash equilibrium given by \eqref{eqn:g2dagger}. First, note that $\ex{(Y^i_1)^2} = 2, i=1,2$ and $\ex{Y^1_1Y^2_1} = 1$. This implies
\beqq{(m^0_2)^2 = \frac{2}{9(22\lambda+59)^2}\left((-8\lambda+44)^2+(16\lambda+32)^2+(-8\lambda+44)(16\lambda+32)\right).}
The expected costs to the controllers at Nash equilibrium $(g^{1\dagger}_1,(g^{2\dagger}_1,g^{2\dagger}_2))$ are given by
\beqq{J^{1\dagger}(\lambda):=J^1(g^{1\dagger}_1,(g^{2\dagger}_1,g^{2\dagger}_2)) &=& \frac{1}{4}(m^0_2)^2+\frac{2(10\lambda +5)^2}{(22\lambda +59)^2}+\frac{37}{21},\\
J^{2\dagger}(\lambda):=J^2(g^{1\dagger}_1,(g^{2\dagger}_1,g^{2\dagger}_2)) & = & \frac{1}{2}(m^0_2)^2+\frac{2(2\lambda+9)^2}{(22\lambda+59)^2}+\frac{41}{21}.}
If we take the limit $\lambda\rightarrow\infty$, we get
\beqq{\lim_{\lambda\rightarrow\infty}J^{1\dagger}(\lambda) = \frac{200+\frac{32}{3}}{484}+\frac{37}{21} \approx 2.197,\quad\lim_{\lambda\rightarrow\infty}J^{2\dagger}(\lambda) = \frac{8+\frac{64}{3}}{484}+\frac{41}{21} \approx 2.013.} 
On the other hand, $J^{1\dagger}(0) \approx 1.832$ and $J^{2\dagger}(0) \approx 2.092$, which corresponds to the expected costs to the controllers if they play according to the common information based Markov perfect equilibrium. Clearly, Controller 2, by choosing an appropriate (very large) value of $\lambda$, can reduce its expected cost, while increasing the expected cost to Controller 1; this observation has connections to incentive designs where Controller 2 can be viewed as the designer (leader) in the game.

\section{Discussion}\label{sec:discussion}
One of the crucial assumptions we made in the game formulation is the strategy independence of beliefs in Assumption \ref{assm:lqgseparation}. As discussed in our companion paper \cite{nayyar2012a}, this assumption plays a crucial role in computing the Markov perfect Nash equilibrium of game {\bf G2} using a backward induction dynamic programming. In this section, we briefly describe the reason why this assumption is important. For a more detailed and technical discussion, we refer the reader to \cite{nayyar2012a}.

Recall that as a consequence of Assumption \ref{assm:lqginfoevolution}, we have
\beqq{\Pi_{t+1} = F_t(\Pi_t,\Gamma^1_t,\Gamma^2_t,\VEC Z_{t+1}).}
Suppose Assumption \ref{assm:lqgseparation} does not hold for some game {\bf G1}, and the common information based conditional belief is dependent on the strategy of Controller 1 at some time step $t_0$. At any time step $t\geq t_0$, in order for Controller 2 to know the belief $\Pi_t$ exactly, the virtual player 1 needs to share its prescriptions in game {\bf G2}. If virtual player 1 does not share its prescription, then virtual player 2 needs to know the precise strategy of virtual player 1. Now, if the virtual players neither share their prescriptions nor their strategies, then each controller has an incentive to deviate from the Nash equilibrium to reap the benefit of asymmetry in the beliefs caused by changing its strategy. In other words, each controller has an incentive to ``deceive'' the other controller if Assumption \ref{assm:lqgseparation} does not hold. Assumption \ref{assm:lqgseparation} assumes that the common information is ``rich enough'' so that controllers cannot deceive each other.

In case the cost functions of the controllers are aligned (that is, they are same at all time steps), then the game problem is just a team problem, and Assumption \ref{assm:lqgseparation} is not required. In team problems, the agents can agree, prior to the start of the play, on what strategies they will use during the course of the play. Moreover, since the cost functions of all the agents are aligned, no agent has an incentive to deviate from the pre-agreed strategies. Consequently, we need not make Assumption \ref{assm:lqgseparation}. 
% This is discussed in greater detail in \cite{nayyar2012a}.

We now consider the case when the cost functions have opposite signs in game {\bf G1}, that is $c^1_t = -c^2_t$ for all time steps $t\in\{1,\ldots, T\}$. This makes game {\bf G1} as a zero-sum game with asymmetric information among the controllers. In this case, if a common information based Markov perfect equilibrium exists, then due to the ordered-interchangeability property of multiple saddle-point equilibria of zero-sum games \cite{Basarbook}, the expected costs to the controllers remain the same for {\it all} other saddle-point equilibrium strategies of the game. Thus, our backward induction algorithm in Theorem \ref{thm:mpevirtual} provides a constructive method to compute saddle-point equilibrium in a zero-sum game of asymmetric information among the controllers, provided that such an equilibrium exists. 

\section{Conclusion and Future Work}\label{sec:conclude}
We studied dynamic two-player linear-Gaussian non-zero sum games in this paper, where the controllers have asymmetric information and their information structures satisfy two assumptions. We showed that under certain structural assumptions on the admissible strategies and stagewise additive cost functions of the players, the existence of a common information based Markov perfect equilibrium can be established by proving the existence of Nash equilibrium in a sequence of static games of symmetric information. For LQG games with cost functions satisfying certain assumptions, we showed that there exists a unique common information based Markov perfect equilibrium. We also gave analytical expressions for the common information based Markov perfect equilibrium for a class of LQG games.

The main idea consisted of defining a new game of symmetric information and perfect observations among virtual players, computing the Markov perfect equilibrium of that game, and then using the Markov perfect equilibrium strategies to obtain a Nash equilibrium of the original game of asymmetric information. We also developed a backward induction algorithm that computes the common information based Markov perfect equilibrium of the game, as long as a Nash equilibrium exists in the static symmetric information game between the virtual players at every time step in the backward induction algorithm. This conceptual approach can be extended to non-zero sum dynamic stochastic games with a finite number of players, which satisfy the two assumptions on the information structures and the assumptions on the cost functions and admissible strategies.

We showed that there may be other Nash equilibria in games, which however cannot be computed using the common information based approach. This is due to the fact that Markov perfect equilibria of the corresponding symmetric and perfect information game among virtual players is a small subclass of Nash equilibria of that game. The conceptual framework we developed in this paper can be used to compute the Nash equilibrium of several classes of games of asymmetric information with infinite and uncountable state and action spaces.

The case of multi-player LQG games can also be solved along similar lines using the result of \cite{Basarmulti}. What still remains to be investigated is the value of information in games with asymmetric information. This is a challenging problem, because as we showed in Section \ref{sec:information}, extra information given to one player generates several other Nash equilibria in the game. It will be interesting to identify a refinement concept for several Nash equilibria arising out of asymmetry in information or extra information to one player.

\section*{Acknowledgments}
This work was supported in part by the AFOSR MURI Grant FA9550-10-1-0573.
% \section*{Appendix}

\appendix

\section{Proof of Lemma \ref{lem:piupdate}}\label{app:piupdate}
Let $\VEC c_t$ be the realized common information at time step $t$. For $i\in\{1,2\}$, let $\gamma^i_t$ be such that $g^i_t(\VEC p^i_t,\VEC c_t)=\gamma^i_t(\VEC p^i_t)$ for all realizations of $\VEC p^i_t\in\ALP P^i_t$. Let $\pi_t(d\VEC s_t) = \mathds{P}\{d\VEC s_t|\VEC c_t\}$. Recall maps $\xi^i_{t+1}, i=1,2$ and $\zeta_{t+1}$ from Assumption \ref{assm:lqginfoevolution}. Let $\SF S_{t+1}\subset\ALP S_{t+1}$ and $\SF Z_{t+1}\subset\ALP Z_{t+1}$ be Borel sets. Now, notice that
\beqq{& &\pr{\SF S_{t+1}\times \SF Z_{t+1}|\VEC c_t} \\
&=& \int_{\SF S_{t+1}\times \SF Z_{t+1}\times \ALP U^{1:2}_t\times\ALP Y^{1:2}_t\times\ALP S_t}
 \pr{d\VEC s_{t+1},d\VEC z_{t+1},d\VEC u^{1:2}_{t},d\VEC y^{1:2}_t,d\VEC s_t|\VEC c_t}\\
 &=& \int_{\SF S_{t+1}\times \SF Z_{t+1}\times \ALP U^{1:2}_t\times\ALP Y^{1:2}_t\times\ALP S_t}
\mathds{1}_{\{\zeta_{t+1}(\VEC p^1_t, \VEC p^2_t, \VEC u^1_t, \VEC u^2_t, \VEC y^1_{t+1}, \VEC y^2_{t+1})\}}(d\VEC z_{t+1})\\
& & \mathds{1}_{\{\xi^1_{t+1}(\VEC p^1_t, \VEC u^1_t,  \VEC y^1_{t+1})\}}(d\VEC p^1_{t+1})\mathds{1}_{\{\xi^2_{t+1}(\VEC p^2_t, \VEC u^2_t,  \VEC y^2_{t+1})\}}(d\VEC p^2_{t+1})\mathds{1}_{\{\gamma^1_t(\VEC p^1_t)\}}(d\VEC u^1_t)\\
& & \mathds{1}_{\{\gamma^2_t(\VEC p^2_t)\}}(d\VEC u^2_t)\mathds{P}\{d\VEC y^1_{t+1},d\VEC y^2_{t+1}|\VEC x_{t+1}\}\mathds{P}\{d\VEC x_{t+1}|\VEC x_t, \VEC u^1_t, \VEC u^2_t\}
\pi_t(d\VEC x_t, d\VEC p^1_t, d\VEC p^2_t).}
It should be noted that the right side of the expression above depends only on $\pi_t$ and the choice of prescription pair $(\gamma^1_t,\gamma^2_t)$. Therefore, if the conditional probability measure $\pr{\SF S_{t+1}\times \SF Z_{t+1}|\VEC c_t}$ is disintegrated with respect to the random variable $\VEC z_{t+1}$, then we get that $\pi_{t+1}(d\VEC s_{t+1}):=\pr{d\VEC s_{t+1}|\VEC c_t,\VEC z_{t+1}}$ depends on $\pi_t$, the choice of prescription pair $(\gamma^1_t,\gamma^2_t)$ and the realization of the random variable $\VEC z_{t+1}$. It should be noted that the measure update equation is a combination of integral equation and a disintegration of probability measure, which does not depend on the choice of the strategy pair $(\VEC g^1,\VEC g^2)$. This completes the proof of the lemma.

\section{Proof of Lemma \ref{lem:pitgaussian}}\label{app:pitgaussian}
The proof is divided into three steps.

{\it Step 1:} At the first time step, since $\VEC X_1, \VEC W^{0:2}_1$ are mutually independent Gaussian random variables and observations are affine functions of the state, we conclude that the joint measure $\pr{d\VEC s_1|\VEC c_1}$ admits a distribution with Gaussian density function. 

{\it Step 2:} For time steps $t\geq 2$, assume that all control laws used till that time are affine functions of common and private information. Moreover, since all noise variables have full support, every possible value of common information $\VEC c_t$ in $\ALP C_t$ can be realized with the choice of affine control laws of the controllers. With affine control laws, the state, the private information and the common information random variables are jointly Gaussian and hence the conditional distribution on the state and private information given $\VEC c_t$ admits a Gaussian density for every $\VEC c_t\in\ALP C_t$.

{\it Step 3:} Assumption \ref{assm:lqgseparation} states that the common information based conditional measure $\pi_t$ does not depend on the choice of control laws. Therefore, under any choice of control laws, the conditional probability measure on the state and the private information given $\VEC c_t$ must admit a Gaussian density for every possible realization $\VEC c_t\in\ALP C_t$. 

Thus, for any  strategy profile of the controllers, the conditional measure $\pi_t$ admits a Gaussian density at all time steps $t\in\{1,\ldots,T-1\}$. This completes the proof of the lemma.

\section{Proof of Lemma \ref{lemma:lqgmean}}\label{app:lqgmean}
Since any Gaussian distribution is characterized by its mean and covariance, from \eqref{eqn:pit1}, we know that
\beq{\label{eqn:pit3}(\VEC M_{t+1},\Sigma_{t+1}) = F_t((\VEC M_t,\Sigma_t),\VEC Z_{t+1}).}
To establish the result, we need to prove that $\Sigma_{t}$ does not depend on the realizations of the random variable $\VEC C_t$, and $F^1_t$ is affine.

We first show that $\Sigma_{t}$ does not depend on the realizations of the random variable $\VEC C_t$. Assume that the control laws of both controllers at all time steps are affine in their information at that time step. Due to linearity of system dynamics and the observation equations, the state, private informations and the common information are jointly Gaussian. Recall that if $(\VEC X,\VEC Y)$ are jointly Gaussian random variables, then the conditional measure on $\VEC X$ given $\VEC y$, denoted by $\pr{d\VEC x|\VEC y}$, admits a Gaussian density function with conditional covariance matrix independent of the realization $\VEC y$. As a consequence of this result, we get that conditional measure on $\ALP S_t = \ALP X_t\times\ALP P^1_t\times\ALP P^2_t$ given the common information $\VEC c_t$ is a Gaussian distribution with the conditional covariance matrix independent of the realization of the common information. Thus, the covariance matrix $\Sigma_t$ evolves according to \eqref{eq:lqgevolution2}.

Now, notice that $\ex{\VEC S_t|\VEC c_t}$ is an affine function of $\VEC c_t$. Therefore, $\VEC M_{t+1}$ is an affine function of $\VEC C_{t+1}$ and $\VEC M_t$ is an affine function of $\VEC C_t$. Combining this with \eqref{eqn:pit3}, we conclude that $\VEC M_{t+1}$ is an affine function of $\VEC M_t$ and $\VEC Z_{t+1}$. Thus, for any time step $t$, $F^1_t$ is an affine function of its arguments, and $\VEC M_t$ evolves according to \eqref{eq:lqgevolution1}. This completes the proof of the lemma.

\section{Proof of Lemma \ref{lem:vpcost}}\label{app:vpcost}
% For a random variable $\VEC C$, let $\mathds{E}_{\VEC C}[\cdot]$ denotes the expectation taken with respect to the probability measure on the random variable $\VEC C$.

We use nested expectation to prove this result. Let $\Gamma^i_t = \chi^i_t(\VEC C_t)$. This gives us
\beqq{\ex{c^i_t(\VEC X_t,\VEC U^1_t,\VEC U^2_t)} &=& \mathds{E}\Big[\ex{c^i_t(\VEC X_t,\VEC U^1_t,\VEC U^2_t)\Big|\VEC C_t}\Big]\\
&=& \mathds{E}\Big[\ex{c^i_t(\VEC X_t,\Gamma^1_t(\VEC P^1_t),\Gamma^2_t(\VEC P^2_t))\Big|\VEC C_t}\Big],\\
&=& \mathds{E}\Big[\tilde c^i_t(\VEC M_t,\chi^1_t(\VEC C_t),\chi^2_t(\VEC C_t))\Big],}
where the first equality follows from the property of expectation, the second equality merely substitutes $\VEC U^i_t = \Gamma^i_t(\VEC P^i_t)$, and the third equality follows from the definition of $\tilde c^i_t$. The above equalities, together with Assumption \ref{assm:lqginfoevolution} (the common information $\VEC C_t$ always increases at all time steps $t$), lead us to the following equalities
\beqq{J^i(g^1,g^2) &=&  \mathds{E}\bigg[\mathds{E}\bigg[\ldots\mathds{E}\bigg[\mathds{E}\bigg[c^i_T(\VEC X_T) \\
& & +c^i_{T-1}(\VEC X_{T-1},g^1_{T-1}(\VEC P^1_{T-1},\VEC C_{T-1}),g^2_{T-1}(\VEC P^2_{T-1},\VEC C_{T-1}))\bigg|\VEC C_{T-1}\bigg]\\
& & +c^i_{T-2}(\VEC X_{T-2},g^1_{T-2}(\VEC P^1_{T-2},\VEC C_{T-2}),g^2_{T-2}(\VEC P^2_{T-2},\VEC C_{T-2}))\bigg|\VEC C_{T-2}\bigg]\ldots\bigg|\VEC C_{1}\bigg]\bigg],\\
&=&  \mathds{E}\bigg[\mathds{E}\bigg[\ldots\mathds{E}\bigg[\mathds{E}\bigg[\tilde c^i_T(\VEC M_T)\\
& & +\tilde c^i_{T-1}(\VEC M_{T-1},\chi^1_{T-1}(\VEC C_{T-1}),\chi^2_{T-1}(\VEC C_{T-1}))\bigg|\VEC C_{T-1}\bigg]\\
& & +\tilde c^i_{T-2}(\VEC M_{T-2},\chi^1_{T-2}(\VEC C_{T-2}),\chi^2_{T-2}(\VEC C_{T-2}))\bigg|\VEC C_{T-2}\bigg]\ldots\bigg|\VEC C_{1}\bigg]\bigg],\\
&=& \tilde J^i(\chi^1,\chi^2).}
This completes the proof of the lemma. The converse can also be proved using similar arguments.

\section{Proof of Lemma \ref{lemma:lqgmarkov}}\label{app:lqgmarkov}
Consider a realization of common information $\VEC c_t$ and realizations $(\VEC m_{1:t},\gamma^{1:2}_{1:t})$ of means and prescriptions until time step $t$. From \eqref{eq:lqgevolution1} in Lemma \ref{lemma:lqgmean}, we have $\VEC M_{t+1} = F^1_t({\VEC m}_t,\VEC Z_{t+1})$. As a consequence of this equation, it is sufficient to prove that 
\beqq{\mathds{P}\{\SF Z_{t+1}|\VEC c_t, \VEC m_{1:t},\gamma^{1:2}_{1:t}\} = \mathds{P}\{\SF Z_{t+1}|\VEC m_{t},\gamma^{1:2}_t\} \text{ for all Borel sets } \SF Z_{t+1}\subset\ALP Z_{t+1}.}
Consider an arbitrary Borel set $\SF Z_{t+1}\subset\ALP Z_{t+1}$. From Assumption \ref{assm:lqginfoevolution}, we get 
\beqq{\VEC Z_{t+1} &=&  \zeta_{t+1}(\VEC P^1_t, \VEC P^2_t, \VEC U^1_t, \VEC U^2_t, \VEC Y^1_{t+1}, \VEC Y^2_{t+1}) \\
&=&  \zeta_{t+1}(\VEC P^1_t, \VEC P^2_t,\gamma^1_t(\VEC P^1_t), \gamma^2_t(\VEC P^2_t), \VEC Y^1_{t+1}, \VEC Y^2_{t+1}),}
where we used the fact that the players used the strategies prescribed by the virtual players. Define $\ALP O_t := \ALP P^1_t\times \ALP P^2_t\times \ALP Y^1_{t+1}\times \ALP Y^2_{t+1} $ and 
\beqq{\tilde\zeta_{t+1}(\VEC O_t,\gamma^1_t,\gamma^2_t):=\zeta_{t+1}(\VEC P^1_t, \VEC P^2_t, \gamma^1_t(\VEC P^1_t), \gamma^2_t(\VEC P^2_t), \VEC Y^1_{t+1}, \VEC Y^2_{t+1}),}
where $\VEC O_t \in\ALP O_t$. Let $N(\cdot;\Sigma)$ denote the density function of a multi-variate mean-zero Gaussian random vector with variance $\Sigma$. Now notice the following:
\beqq{& & \mathds{P}\{\VEC Z_{t+1}\in \SF Z_{t+1}|\VEC c_t, \VEC m_{1:t},\gamma^{1:2}_{1:t}\} \\
&=& \int_{\SF Z_{t+1}}\int_{\ALP X_t\times \ALP X_{t+1}\times \ALP O_t} \mathds{P}\{d\VEC z_{t+1},d\VEC x_t,d\VEC x_{t+1},d\VEC o_t|\VEC c_t, \VEC m_{1:t},\gamma^{1:2}_{1:t}\}\\
& =&  \int_{\SF Z_{t+1}}\int_{\ALP X_t\times \ALP X_{t+1}\times \ALP O_t}\ind{\tilde\zeta_{t+1}(\VEC o_t,\gamma^1_t,\gamma^2_t)}(d\VEC z_{t+1}) \mathds{P}\{d\VEC y^{1:2}_{t+1}|\VEC x_{t+1}\}\\
& & \mathds{P}\{d\VEC x_{t+1}|\VEC x_t,\gamma^1_t(\VEC p^1_t),\gamma^2_t(\VEC p^2_t)\} \mathds{P}\{d\VEC x_t,d\VEC p^1_t,d\VEC p^2_t|\VEC c_t, \VEC m_{1:t},\gamma^{1:2}_{1:t}\},\\
&=& \int_{\SF Z_{t+1}}\int_{\ALP X_t\times \ALP X_{t+1}\times \ALP O_t}\ind{\tilde\zeta_{t+1}(\VEC o_t,\gamma^1_t,\gamma^2_t)}(d\VEC z_{t+1}) \mathds{P}\{d\VEC y^{1:2}_{t+1}|\VEC x_{t+1}\}\\
& & \mathds{P}\{d\VEC x_{t+1}|\VEC x_t,\gamma^1_t(\VEC p^1_t),\gamma^2_t(\VEC p^2_t)\} N(\VEC s_t-\VEC m_t;\Sigma_t)d\VEC s_t,}
where we used the fact that the conditional distribution $\mathds{P}\{d\VEC x_t,d\VEC p^1_t,d\VEC p^2_t|\VEC c_t, \VEC m_{1:t},\gamma^{1:2}_{1:t}\}$ is a Gaussian distribution with mean $\VEC m_t$ and variance $\Sigma_t$ (recall that $\Sigma_t$ is independent of the realizations of random variables by Lemma \ref{lemma:lqgmean}). The right side of the equation above depends only on $\VEC m_t$ and the choice of prescriptions $\gamma^{1:2}_t$. This establishes the result of the lemma.

\section{Proof of Theorem \ref{thm:lqgequiv}}\label{app:lqgequiv}
Let $(\chi^{1\star}, \chi^{2\star})$ be a Nash equilibrium strategy profile of game \textbf{G2}. We want to show that the strategy profile $(\mathbf g^{1\star},\mathbf g^{2\star})$ is a Nash equilibrium of game {\bf G1}. Let $\mathbf g^1\in\ALP G^1_{1:T-1}$ be an arbitrary strategy of Player 1. Define $\chi^1:=\varsigma^1(\VEC g^1)$. From the definition of Nash equilibrium of game {\bf G2}, we get
\beqq{J^1(\mathbf g^{1\star},\mathbf g^{2\star}) = \tilde J^1(\chi^{1\star}, \chi^{2\star})\leq \tilde J^1(\chi^{1}, \chi^{2\star}) =  J^1(\mathbf g^{1},\mathbf g^{2\star}),}
where we used the result of Lemma \ref{lem:vpcost}. Similarly, we get $J^2(\mathbf g^{1\star},\mathbf g^{2\star})\leq J^2(\mathbf g^{1\star},\mathbf g^{2})$ for all $\mathbf g^2\in\ALP G^2_{1:T-1}$. Thus, strategy profile $(\mathbf g^{1\star},\mathbf g^{2\star})$ is indeed a Nash equilibrium of game {\bf G1}. 

Using similar steps as above, we prove the converse. This establishes the result of the theorem.

\section{Proof of Lemma \ref{lemma:lqginfostatelemma}} \label{app:lqginfostatelemma}
% We first prove Part 1 of the lemma. 
Assume that virtual player $2$ uses a fixed strategy of the form $\Gamma^2_t = \psi^2_t(\VEC M_t), t\in\{1,\ldots, T\}$. We show that the virtual player $1$'s problem is simply a Markov decision problem with state process $\{\VEC M_t\}_{t\in\{1,\ldots,T\}}$ and actions $\{\Gamma^1_t\}_{t\in\{1,\ldots,T-1\}}$. 

Suppose at time $t$, $\VEC c_t$ is the realization of the common information, $\VEC m_t$ is the realization of the mean of the conditional measure $\pi_t$, $\gamma^2_t := \psi^2_t(\VEC m_t)$, and virtual player $1$ selects $\gamma^1_t$ as its action. Recall that $\mathds{P}\{d\VEC x_t,d\VEC p^1_t,d\VEC p^2_t|\VEC c_t\}$ is a Gaussian distribution with mean $\VEC m_t$ and variance $\Sigma_t$. The expected instantaneous cost is
\beqq{\tilde{c}^1_t(\VEC m_t,\gamma^1_t,\psi^2_t(\VEC m_t)) &=& \ex{c^1_t(\VEC X_t,\gamma^1_t(\VEC P^1_t),\gamma^2_t(\VEC P^2_t))|\VEC c_t}\\
&=& \int_{\ALP S_t}c^1_t(\VEC x_t,\gamma^1_t(\VEC p^1_t),\gamma^2_t(\VEC p^2_t)) \mathds{P}\{d\VEC s_t|\VEC c_t\}.}
Since $\mathds{P}\{d\VEC s_t|\VEC c_t\}$ and $\gamma^2_t$ are functions of $\VEC m_t$ at all time steps $t\in\{1,\ldots, T-1\}$, we conclude that the cost of virtual player 1 at time step $t$ is only a function of $\VEC m_t$ and $\gamma^1_t$. Also recall that in Lemma \ref{lemma:lqgmarkov}, we proved that $\{\VEC M_t\}_{t\in\{1,\ldots,T\}}$ is a controlled Markov chain. Therefore, virtual player $1$'s optimization is a Markov decision problem with $\VEC M_t$ as the state and $\gamma^1_t$ as the controlling action. The corresponding statement for virtual player 2 can also be proved using similar arguments. This completes the proof of the lemma.

% We now prove Part 2 of the lemma. Given the sufficient conditions about the Markov decision problem, from the theory of Markov decision processes \cite[Section 3.3]{lerma1996}, we conclude that there exists an optimal measurable function that maps the space $\ALP S_t$ of all means $\VEC M_t$ to the space $\ALP A^1_t$ of prescriptions $\gamma^1_t$. Thus, there exists a set of measurable functions $\{\psi^1_t\}_{t\in\{1,\ldots,T\}}$ that maps $\VEC M_t$ to the action $\gamma^1_t$ of virtual player $1$ which minimizes the one-sided optimization problem of virtual player 1 given the control laws of virtual player 2. 
% 
% A similar statement can be proved for virtual player 2 as well by interchanging the superscripts 1 and 2, and mimicking the same steps as above. This completes the proof of the lemma.

\section{Proof of Theorem \ref{thm:mpevirtual}}\label{app:mpevirtual}
Suppose that the strategy profile $(\psi^{1\star},\psi^{2\star})$ of virtual players is a Markov perfect equilibrium of game {\bf G2}. Fix a time step $t\in\{1,\ldots,T-1\}$ and a virtual player $i\in\{1,2\}$. By definition of Markov perfect equilibrium, we know that $\psi^{i\star}_{t:T-1}$ minimizes the expected cost
\beqq{\ex{\tilde c^i_T(\VEC M_T)+\sum_{s=t}^{T-1}\tilde{c}^i_s(\VEC M_s,\gamma^i_s,\psi^{-i\star}_s(\VEC M_s))}.}
Applying the principle of dynamic programming, we can easily verify that \eqref{eqn:psiitm} is satisfied by the control law $\psi^{i\star}_t$.
 
We now prove the converse. For a fixed sub-game strategy $\psi^j_{t:T-1}$ of virtual player $j$, let us denote the one-sided optimization problem for virtual player $i\neq j$ at time instant $t$ given the Markov state $\VEC m_t$ by $MDP^i_t(\psi^j_{t:T-1},\VEC m_t)$. We prove the converse of the statement of the theorem by showing that the strategy pair $(\psi^{1\star}_{t:T-1},\psi^{2\star}_{t:T-1})$ as defined by \eqref{eqn:psiitm} is a sub-game perfect equilibrium for every sub-game starting at time instant $t\in\{1,\ldots,T-1\}$.

Fix the strategy profile of virtual player 2 to $\psi^{2\star} = \{\psi^{2\star}_1,\ldots,\psi^{2\star}_{T-1}\}$ and fix any time step $t\in\{1,\ldots,T-1\}$.
% Consider a realization of the Markov state $\VEC m_t$. Then, the virtual player 1's optimization problem is $MDP^1_t(\psi^{2\star}_{t:T-1},\VEC m_t)$. First note that due to dynamic programming recursion, the optimization problem that virtual player $i$ faces at any time step $t\in\{1,\ldots,T-1\}$ is
% \beqq{\underset{\tilde\gamma^i_t\in \ALP A^i_t}{\min} \;\;\ex{\tilde c^1_T(\VEC M_T)+\sum_{s=t+1}^{T-1}\tilde{c}^1_s(\VEC M_s,\psi^{1\star}_s(\VEC M_s),\psi^{2\star}_s(\VEC M_s))\bigg|\VEC M_t =\VEC m_t}\\
% +\tilde{c}^1_t(\VEC m_t,\tilde\gamma^1_t,\psi^{2\star}_t(\VEC m_t)).}
% In the above equation, note that the quantity 
% \beqq{& & \ex{\tilde c^1_T(\VEC M_T)+\sum_{s=t+1}^{T-1}\tilde{c}^1_s(\VEC M_s,\psi^{1\star}_s(\VEC M_s),\psi^{2\star}_s(\VEC M_s))\Big|\VEC m_t}\\
% & & =\ex{\ex{\tilde c^1_T(\VEC M_T)+\sum_{s=t+1}^{T-1}\tilde{c}^1_s(\VEC M_s,\psi^{1\star}_s(\VEC M_s),\psi^{2\star}_s(\VEC M_s))\Big|\VEC M_{t+1}}\Bigg|\VEC m_t},}
% which is precisely $\ex{V^1_{t+1}(\VEC M_{t+1})\big|\VEC m_t}$ by the property of iterated expectation. Lemma \ref{lemma:lqginfostatelemma} implies that the minimum strategy of the one-sided optimization problem is just a function of conditional mean. 
Then, using the principle of dynamic programming for Markov decision processes, the recursion in \eqref{eqn:viT} and \eqref{eqn:vit} implies that $\psi^{1\star}_{t:T-1}$ is the optimal strategy of virtual player 1 for $MDP^1_t(\psi^{2\star}_{t:T-1},\VEC m_t)$. Since the time step $t$ was arbitrary, we conclude that $\psi^{1\star}$ is the best response strategy of virtual player 1. A similar argument proves that $\psi^{2\star}$ is the best response strategy of virtual player 2 given the virtual player 1's strategy $\psi^{1\star}$. This completes the proof of the converse part of the theorem.

\section{Proof of Lemma \ref{lem:basargame}}\label{sec:auxiliary}
We first need several results about eigenvalues, eigenvectors, and pseudo-inverses of symmetric non-invertible matrices. We turn our attention to the proof of Lemma \ref{lem:basargame} thereafter.
\begin{lemma}\label{lem:invmat}
In the statements below, for any square matrix $D\in \Re^{n\times n}$, $D^{-1}$ denotes generalized inverse of the matrix.
\begin{enumerate}
\item For any square matrix $D\in \Re^{n\times n}$, $DD^{-1}D=D$.
\item Let $D_1$ and $D_2$ be matrices of appropriate dimensions. Then, the sets of non-zero eigenvalues of $D_1D_2$ and $D_2D_1$ are the same.
\item Any symmetric matrix $D\in \Re^{n\times n}$ has real eigenvalues and a mutually orthogonal set of eigenvectors that spans $\Re^n$. Thus, spectral radius of $D^\transpose D = D^2$ is the same as the square of the spectral radius of $D$. 
\item For any matrix $D\in\Re^{m\times n}$, the matrix $D^\transpose(DD^\transpose)^{-1}D$ is a symmetric positive semi-definite matrix with every non-zero eigenvalue of the matrix equal to 1.
\item Let $D_1,D_2$ be matrices of appropriate dimensions. Define a square matrix $D$ as \beqq{D:=(D_1D_1^\transpose)^{-1}D_1D_2^\transpose(D_2D_2^\transpose)^{-1}D_2D_1^\transpose.}
Then, all eigenvalues of $D$ are non-negative real numbers that are less than or equal to 1.
\end{enumerate}
\end{lemma}
\begin{proof}
\begin{enumerate}
\item See \cite[Theorem 4.3.2, p. 100]{catlin1989}
\item See \cite[p. 24]{minc1964}.
\item Since $D$ is symmetric, it is normal, that is, it is diagonalizable via a similarity transformation \cite[4.10.3, p. 67]{minc1964}. From \cite[4.10.3, p. 67]{minc1964}, we also know that the set of eigenvectors of a symmetric matrix in $\Re^{n\times n}$ forms an orthonormal basis for $\Re^n$ (for a proof, see \cite[Theorem 2, p. 54]{bellman1970}).
\item The fact that  $D^\transpose(DD^\transpose)^{-1}D$ is symmetric is clear. We first prove that $D^\transpose(DD^\transpose)^{-1}D$ is positive semi-definite. Since $DD^\transpose$ is positive semi-definite, $(DD^\transpose)^{-1}$ is positive semi-definite by \cite[Theorem 4.4.3, p. 109]{minc1964}. Thus, $D^\transpose(DD^\transpose)^{-1}D$ is positive semi-definite, and therefore, its eigenvalues are non-negative real numbers.  Using Part 2, we know that $D^\transpose(DD^\transpose)^{-1}D$ and $(DD^\transpose)^{-1}DD^\transpose$ have the same set of non-zero eigenvalues. Now, let $\lambda$ be a non-zero eigenvalue of $(DD^\transpose)^{-1}DD^\transpose$ and $\VEC e$ be the corresponding eigenvector. Then,
\beqq{(DD^\transpose)^{-1}DD^\transpose\VEC e = \lambda \VEC e.}
From the above equation, we note that $\VEC e$ cannot lie in the nullspace of $D^\transpose$ as $\lambda\neq 0$, which further implies that $ \VEC e^\transpose DD^\transpose\VEC e >0$. Multiplying both sides  in the above equation by $\VEC e^\transpose DD^\transpose$ and using the identity in Part 1, we get
\beqq{\VEC e^\transpose DD^\transpose(DD^\transpose)^{-1}DD^\transpose\VEC e = \VEC e^\transpose DD^\transpose\VEC e = \lambda \VEC e^\transpose DD^\transpose\VEC e.}
Thus, $\lambda = 1$. This completes the proof of this part of the lemma.

\item Let $P_i := D_i^\transpose(D_iD_i^\transpose)^{-1}D_i$ for $i=1,2$. Note that $P_1$ and $P_2$ are symmetric positive semi-definite matrices and have the same dimension. Part 2 of the lemma implies that $D$ and $P_1P_2$ have the same set of non-zero eigenvalues. Part 4 of the lemma implies that spectral radius of $P_i$ is equal to 1 for $i=1,2$. Part 3 implies $P_i^\transpose P_i$ has spectral radius 1 for $i=1,2$. Let $\lambda$ be a non-zero eigenvalue of $P_1P_2$ and $\VEC e$ be the corresponding normalized eigenvector. Then, we have $P_1P_2 \VEC e = \lambda \VEC e$. Taking the usual norm on both sides of the equation, we get 
\beqq{|\lambda|^2 \VEC e^\transpose\VEC e  = ( P_2\VEC e)^\transpose P_1^\transpose P_1 (P_2 \VEC e)\leq \VEC e^\transpose P_2^\transpose P_2 \VEC e\leq \VEC e^\transpose\VEC e.}
Thus, $|\lambda|\leq 1$, which completes the proof of this part of the lemma.
\end{enumerate}
\end{proof}

\begin{lemma}\label{lem:p1p2p3}
Let $\lambda(\cdot)$ denote the spectral radius of a matrix $(\cdot)$. Let $P_1, P_2$ and $P_3$ be square matrices in $\Re^{n\times n}$. If $\lambda(P_1)\lambda(P_2)< 1$, then there exists a unique $D\in \Re^{n\times n}$ such that $D+P_1DP_2 = P_3$ is satisfied.
\end{lemma}
\begin{proof}
See the proof of Theorem 3 in \cite{Basar1975}.
\end{proof}

We now prove Lemma \ref{lem:basargame} in three steps. First, we consider another game {\bf AG2} with a different cost function and show that games {\bf AG1} and {\bf AG2} have the same set of Nash equilibria. The cost functions of the players in game {\bf AG2} have a form that is similar to the game considered in \cite{Basar1975}. However, in \cite{Basar1975}, the matrices $\Sigma_{y^iy^i}, i=1,2$  are assumed to be invertible, which we relax in this proof. In Step 2 of the proof, we use the result from \cite{Basar1975,Basaronestep,Basarmulti} to show that the Nash equilibrium of game {\bf AG2} exists and is affine in the information of the players. Then, we use the Step 1 of the proof to obtain the Nash equilibrium of game {\bf AG1}.

{\it Step 1:} Consider game {\bf AG2} in which the players have the following cost functions
\beq{\label{eqn:bj1aux}\bar c^1(\VSF X,\VSF U^1,\VSF U^2) &=& \VSF U^{1\transpose}C_{22}\VSF U^1+ 2\VSF U^{1\transpose}C_{12}^\transpose \VSF X+2\VSF U^{1\transpose}C_{23}\VSF U^2 +2d_2 \VSF U^1,\\
\label{eqn:bj2aux}\bar c^2(\VSF X,\VSF U^1,\VSF U^2) &=& \VSF U^{2\transpose}E_{33}\VSF U^2+2\VSF U^{2\transpose}E_{13}^\transpose \VSF X+2\VSF U^{2\transpose}E_{23}^\transpose\VSF U^1+2f_3 \VSF U^2.}
The difference in {\bf AG1} and {\bf AG2} lies in the cost functions of the players. In game {\bf AG1}, cost $c^i$ has the terms that are not dependent on $\VSF U^i$ for $i=1,2$ whereas in game {\bf AG2}, cost $\bar c^i$ has only the terms dependent on $\VSF U^i$ for $i=1,2$. Thus, games {\bf AG1} and {\bf AG2} are strategically equivalent. 

{\it Step 2:} In this step, we prove that game {\bf AG2} has a unique Nash equilibrium that is affine in the information of the controllers. From \cite[Theorem 1, p. 236]{Basarmulti}, we know that if Assumption \ref{assm:eigcost} holds, then the Nash equilibrium strategy tuple of game {\bf AG2} exists, is unique and affine in its argument. Assume that the Nash equilibrium strategies are given by
\beqq{g^{i\star}(\VSF Y^i) =  T^i(\VSF Y^i-\VSF m_{y^i}) + b^i ,\qquad i=1,2.}
Then, $b^1,b^2$ must be the solutions of the following pair of equations
\beqq{b^1 &=& -C_{22}^{-1}[d_2^{\transpose}+C_{12}\VSF m_x+C_{23}b^2] \\
b^2 &=& -E_{33}^{-1}[f_3^{\transpose}+E_{13}\VSF m_x +E_{23}^{\transpose}b^1 ],}
and $T^1, T^2$ must be the solutions of the following pair of equations
\beqq{T^1 &=& -C_{22}^{-1}[C_{12}^\transpose\Sigma_{xy^1}\Sigma_{y^1y^1}^{-1}+ C_{23}T^2\Sigma_{y^2y^1}\Sigma_{y^1y^1}^{-1}], \\
T^2 &=& -E_{33}^{-1}[E_{13}^\transpose\Sigma_{xy^2}\Sigma_{y^2y^2}^{-1}+E_{23}^{\transpose}T^1\Sigma_{y^1y^2}\Sigma_{y^2y^2}^{-1}].}
We now show that there exist pairs $(b^1,b^2)$ and $(T^1,T^2)$ which satisfy the above set of equations. If $\Sigma_{y^iy^i}, i=1,2$ are invertible, then the existence of such pairs $(b^1,b^2)$ and $(T^1,T^2)$ follow from \cite{Basar1975}. We now prove that such pairs exist even if $\Sigma_{y^iy^i}, i=1,2$ are not invertible. In what follows, $\Sigma_{y^iy^i}^{-1}$ represents the generalized inverse of $\Sigma_{y^iy^i}, i=1,2$.

Since Assumption \ref{assm:eigcost} holds, there exists an $i_0\in\{1,2\}$ and a matrix $K\in \ALP K_{i_0}$ such that $\bar{\lambda}(K)<1$. Without loss of generality, assume that $i_0=1$ and let $L$ be the matrix such that $K = LK_1L^{-1}$. Let $\tilde b^1 = L b^1, \tilde b^2 = Lb^2,\tilde T^1 = LT^1$ and $\tilde T^2 = LT^2$. Now, notice that 
\beqq{\tilde b^1 &=& -LC_{22}^{-1}[d_2^{\transpose}+C_{12}\VSF m_x+C_{23}L^{-1}\tilde b^2] \\
\tilde b^2 &=& -LE_{33}^{-1}[f_3^{\transpose}+E_{13}\VSF m_x +E_{23}^{\transpose}L^{-1}\tilde b^1],}
which admits a unique solution since $\bar{\lambda}(K)<1$. This further implies that a pair of $(b^1, b^2)$ exists. We now prove that there exist $T^1$ and $T^2$ satisfying the above pair of equations. Substituting the expression for $T^2$ in the expression of $T^1$ and writing the expression in terms of $\tilde T^1$, we get
\beq{\tilde T^1 = -LC_{22}^{-1}C_{12}^\transpose\Sigma_{xy^1}\Sigma_{y^1y^1}^{-1}+ LC_{22}^{-1}C_{23}E_{33}^{-1}E_{13}^\transpose\Sigma_{xy^2}\Sigma_{y^2y^2}^{-1}\nonumber\\
+ L(C_{22}^{-1}C_{23}E_{33}^{-1}E_{23}^{\transpose})L^{-1}\tilde T^1(\Sigma_{y^1y^2}\Sigma_{y^2y^2}^{-1}\Sigma_{y^2y^1}\Sigma_{y^1y^1}^{-1}),\label{eqn:t1unique}}
where $L(C_{22}^{-1}C_{23}E_{33}^{-1}E_{23}^{\transpose})L^{-1}$ is equal to $K$. Also note that by Lemma A2 in \cite[p. 327]{Basar1975}, $\bar{\lambda}(K)<1$ implies $\lambda(K)<1$. Now, recall that $\Sigma_{y^iy^j} = \Sigma_{y^iy^i}^{\frac{1}{2}}\Sigma_{y^jy^j}^{\frac{1}{2}\transpose}$ for $i,j=1,2$. Thus, Lemma \ref{lem:invmat} Part 5 implies that $\lambda(\Sigma_{y^1y^2}\Sigma_{y^2y^2}^{-1}\Sigma_{y^2y^1}\Sigma_{y^1y^1}^{-1})\leq 1$. Since $\lambda(K)\lambda(\Sigma_{y^1y^2}\Sigma_{y^2y^2}^{-1}\Sigma_{y^2y^1}\Sigma_{y^1y^1}^{-1})<1$, we conclude from Lemma \ref{lem:p1p2p3} that there exists a unique $\tilde T^1$, which satisfies \eqref{eqn:t1unique}. This further implies the existence of a unique $T^1$ in \eqref{eqn:t1}-\eqref{eqn:t2}. We can substitute this value of $T^1$ in the expression for $T^2$ to get its unique value.

The case of $i_0=2$ is analogous to the argument as above. We first get a unique value of $T^2$ and then substitute $T^2$ in the expression for $T^1$ to get the unique value of $T^1$. Also note that computing the value of $b^1$, $b^2$, $T^1$ and $T^2$ is equivalent to solving a linear program.

{\it Step 3:} Since games {\bf AG2} and {\bf AG1} are strategically equivalent, we can use the result of Step 2 to obtain the Nash equilibrium strategies of the players for game {\bf AG1}. 

The proof of Part 2 of the lemma is thus complete.

\bibliographystyle{IEEEtran}
\bibliography{myref,collection,papers,mypaper,game,refbib}

% \newpage
% \vspace{20pt}
% \begin{center}\textbf{\large{Supplementary Material}}\end{center}

\end{document}

%% file: presiam
\usepackage{graphicx}
\usepackage{amsmath,amssymb,dsfont}
\usepackage{color}
\usepackage[english]{babel}
\usepackage{mathrsfs}
\usepackage[mathscr]{euscript}

\newtheorem{assumption}{Assumption}
\newtheorem{remark}{Remark}

\usepackage[sort,compress]{cite}
\usepackage[cmex10]{mathtools}

\let\ALP \mathcal

\let\VEC      \mathbf

\let\SF \mathscr

\newcommand\transpose{*}

\newcommand\IND{\mathds{1}}

\usepackage[T1]{fontenc}

\usepackage{xspace}
\newcommand{\ind}[1]{\IND_{\{#1\}}}
\newcommand{\beq}[1]{\begin{eqnarray} #1 \end{eqnarray}}
\newcommand{\beqq}[1]{\begin{eqnarray*} #1 \end{eqnarray*}}
\renewcommand{\Re}{\mathds{R}}
\newcommand{\ex}[1]{\mathds{E}\left[#1\right]}
\newcommand{\pr}[1]{\mathds{P}\left\{#1\right\}}
\renewcommand{\matrix}[2]{\left[\begin{array}{#1} #2 \end{array}\right]}